\documentclass{article}[12pt]

\usepackage{comment}  
\usepackage{algorithm}  
\usepackage{algorithmicx}  
\usepackage[noend]{algpseudocode}  
\usepackage{amssymb}
\usepackage{epsfig}
\usepackage{amsthm}
\usepackage{bm}
\usepackage{amsmath}
\usepackage{amsfonts}
\usepackage{color}
\usepackage{multirow}
\usepackage{dsfont}
\usepackage[round]{natbib}
\usepackage{setspace}

\newcommand{\ind}{\mathbb{I}} 
\newcommand{\U}{U_{n,2}}
\newcommand{\field}[1]{\mathbb{#1}}
\newcommand{\R}{\field{R}}
\newcommand{\F}{\field{F}}

\newcommand{\N}{\field{N}}

\newcommand{\E}{\field{E}}

\theoremstyle{example}
\theoremstyle{remark}
\theoremstyle{lemma}
\theoremstyle{definition}
\theoremstyle{corol}
\theoremstyle{proposition}
\theoremstyle{condition}
\theoremstyle{assumption}
\newtheorem{assumption}{\n{Assumption}}[section]
\newtheorem{theorem}{\n{Theorem}}[section]

\newtheorem{remark}{\n{Remark}}[section]
\newtheorem{lemma}{\n{Lemma}}[section]

\newtheorem{proposition}{\n{Proposition}}[section]

\def\lf{\lfloor}
\def\rf{\rfloor}

\font\n=cmcsc10

\def\var{{\mbox{var}}}

\def\diag{{\mbox{diag}}}
\def\cum{{\mbox{cum}}}
\title{Adaptive Inference for Change Points in High-Dimensional Data }
\author{Yangfan Zhang, Runmin Wang and Xiaofeng Shao \footnote{ Yangfan Zhang is Ph.D. student, Xiaofeng Shao is  Professor at Department of Statistics, University of Illinois at Urbana Champaign. Runmin Wang is Assistant Professor at  Department of Statistical Science, Southern Methodist University.   Emails: yangfan3@illinois.edu, runminw@mail.smu.edu and xshao@illinois.edu.  We would like to thank two anonymous referees  for constructive comments, which led to substantial improvements. We are also grateful to  Dr. Farida Enikeeva for sending
us the code used in Enikeeva and Harchaoui (2019). Shao's research is partially supported by NSF-DMS
1807023 and NSF-DMS-2014018. }}
\date{}

\begin{document}

\maketitle

{\bf Abstract}: In this article, we propose a class of test statistics for a change point in the mean of high-dimensional independent data. Our test integrates the U-statistic based approach in a recent work by \cite{hdcp} and the $L_q$-norm based high-dimensional test in \cite{he2018}, and inherits several appealing features such as being tuning parameter free and asymptotic independence for test statistics corresponding to even $q$s.  A simple combination of test statistics corresponding to several different $q$s leads to a test with adaptive power property, that is, it can be powerful against both sparse and dense alternatives. On the estimation front, we obtain the convergence rate of the maximizer of our test statistic  standardized by sample size when there is one change-point in mean and $q=2$, and propose to   combine our tests with a wild binary segmentation (WBS) algorithm to estimate the change-point number and locations when there are  multiple change-points. Numerical comparisons using both simulated and real data demonstrate the advantage of our adaptive test and its corresponding estimation method.


\section{Introduction}
\label{sec:intro}

Testing and estimation of change points in a sequence of time-ordered data is a classical problem in statistics. There is a rich literature for both univariate and multivariate data of low dimension; see \cite{csorgo1997}, \cite{chen2011} and \cite{tar2014}, for some book-length introductions and \cite{pe2006}, \cite{au2013}, and \cite{ami2017} for recent reviews of the subject. This paper addresses the testing and estimation for change points of high-dimensional data where the dimension $p$ is high and can exceed the sample size $n$.

As high-dimensional data becomes ubiquitous due to technological advances in science, engineering and other areas, change point inference under the high-dimensional setting has drawn great interest in recent years. When the dimension $p$ is greater than sample size $n$, traditional methods are often no longer applicable. Among recent work that addresses change point inference for the mean of high-dimensional data, we mention \cite{ho2012}, \cite{chan2013}, \cite{ji2015}, \cite{cho2016},  \cite{wang2018},  \cite{en2013},  \cite{hdcp} and \cite{yu20}. In most of these papers, the proposed methods are powerful either when the alternative is sparse and strong, i.e., there are a few large non-zero values in the components of mean difference, or when the alternative is weak and dense, i.e., there are many small values in the components of mean difference. 
Among some of these papers, the sparsity appeared either explicitly in the assumptions, e.g. \cite{wang2018}, who proposed to project the data to some informative direction related to the mean change, to which univariate change point detection algorithm can be applied,  or implicitly in the methodology, e.g. \cite{ji2015}, who took the maximal CUSUM statistic and therefore essentially targeted at the sparse alternative. \cite{yu20} recently introduced a Gaussian multiplier bootstrap to calibrate critical values of the sup norm of CUSUM test statistics in high dimensions and their test is also specifically for sparse alternative.
On the contrary, \cite{ho2012} aggregated the univariate CUSUM test statistics using the sum and their test is supposed to capture the dense alternative, but the validity of their method required the cross-sectional independence assumption. 
\cite{hdcp} aimed at dense alternatives by extending the U-statistic based approach pioneered by \cite{chenqin2010} in the two-sample testing problem.
An exception is the test developed in \cite{en2013}, which was based on a combination of a linear statistic and a scan statistic, and can be adaptive to both sparse and dense alternatives. However, its critical values were obtained under strong Gaussian and independent components assumptions and they do not seem to work when these assumptions are not satisfied; see Section~\ref{sec:simulation} for numerical evidence. 

In practice, it is often unrealistic to assume a particular type of alternative and there is little knowledge about the type of changes if any. Thus there is a need to develop new test that can be adaptive to different types of alternatives, and have good power against a broad range of alternatives. In this article, we shall propose a new class of tests that can have this adaptive power property, which holds without the strong Gaussian and independent components assumptions. 
Our test is built on two recent advances in the high-dimensional testing literature: \cite{hdcp} and \cite{he2018}. In \cite{hdcp}, they developed a mean change point test based on a U-statistic that is an unbiased estimator of the squared $L_2$ norm of the mean difference. They further used the idea of self-normalization [see \cite{shao2010a}, \cite{shao2010}, \cite{shao2015}] to eliminate the need of estimating the unknown nuisance parameter. \cite{he2018} studied both one sample and two sample high-dimensional testing problem for the mean and covariance matrix using $L_q$ norm, where $q\in [2,\infty]$ is some integer. They showed that the corresponding U-statistics at different $q$s are asymptotically independent, which facilitates a simple combination of the tests based on several values of $q$ (say $2$ and $\infty$) and their corresponding $p$-values, and that the resulting combined test is adaptive to both dense and sparse alternatives.

Building on these two recent advances, we shall propose a new $L_q$ norm based test for a change point in the mean of high-dimensional independent data. Our contributions to the literature is threefold. 
On the methodological front, we develop a new class of test statistics (as indexed by $q\in 2\N$) based on the principle of self-normalization in the high-dimensional setting. Our test is tuning parameter free when testing for a single change point. A simple combination of tests corresponding to different $q$s can be easily implemented due to the asymptotic independence and results in an adaptive test that has well-rounded power against a wide range of alternatives.  
On the theory front, as \cite{he2018} proved the asymptotic independence of one-sample and two-sample U-statistics corresponding to different $q$s, we  derive the asymptotic independence for several stochastic processes corresponding to different $q$s under significantly weaker assumptions. More precisely, we can define two-sample test statistics on different sub-samples for each $q\in 2\N$. These statistics can be viewed as smooth functionals of stochastic processes indexed by the starting and ending points of the sub-samples, which turn out to be asymptotically independent for different $q$s.
Compared to the adaptive test in \cite{en2013}, which relied on the Gaussian and independent components assumptions, our technical assumptions are much weaker, allowing non-Gaussianity and weak dependence among components. Furthermore, we obtained the convergence rate of the argmax of our SN-based test statistic  standardized by sample size when there is one change point and $q=2$. 
Lastly, in terms of empirical performance, we show in the simulation studies that the adaptive test can have accurate size  and high power for both sparse and dense alternatives. Their power is always close to the highest one given by a single statistic under both dense and sparse alternatives. 

The rest of the paper is organized as follows. In Section 2, we define our statistic, derive the limiting null distribution and analyze the asymptotic power when there is one change point. We also propose an adaptive procedure combining several tests of different $q\in 2\N$. In Section 3, we study the asymptotic behavior of change-point location estimators when there is a single change-point and combine the WBS algorithm with our test to estimate the location when there are multiple change points. In Section 4, we present some simulation results for both testing and estimation and apply the WBS-based estimation method to a real data set. Section 5 concludes. All technical details and some additional simulation results are gathered in the supplemental material.






\section{Test Statistics and Theoretical Properties}
\label{sec:test}

 Mathematically, let  $\{Z_t\}_{t=1}^n\in\R^p$ be i.i.d random vectors with mean 0 and covariance $\Sigma$. Our observed data is $X_t=Z_t+\mu_t$, where $\mu_t=E(X_t)$ is the mean at time $t$. The null hypothesis is that there is no change point in the mean vector $\mu_t$ and the alternative is that there is at least one change point, the location of which is unknown, i.e., we want to test
$$
\mathcal{H}_{0} : \mu_{1}=\mu_{2}=\cdots=\mu_{n} \quad v . s \quad \mathcal{H}_{1} : \mu_{1}=\cdots=\mu_{k_{1}} \neq \mu_{k_{1}+1}=\cdots=\mu_{k_{s}} \neq \mu_{k_{s}+1} \cdots=\mu_{n},
$$
where $k_1<k_2<\cdots<k_s$ and $s$ are unknown. Note that we assume temporal independence, which seems to be commonly adopted in change point analysis for genomic data; see \cite{zhang10}, \cite{jeng10}, and  \cite{zhang12} among others.

In this section, we first construct our two-sample U-statistic for a single change point alternative, which is the cornerstone for the estimation method we will introduce later. Then we derive the theoretical size and power results for our statistic. We also form an adaptive test that combines tests corresponding to different $q$s. Throughout the paper, we assume $p\wedge n\rightarrow+\infty,$ and we may use $p=p_n$ to emphasize that $p$ can depend on $n$. For a vector or matrix $A$ and $q\in2\N$, we use $\|A\|_q$ to denote $\big(\sum_{i,j}A_{ij}^q\big)^{1/q}$, and in particular, for $q=2,\|\cdot\|_q=\|\cdot\|_F$ equals the Frobenius norm. We use $\|\Sigma\|_s$ to denote the spectral norm. Denote the number of permutations $P_{q}^{k}=k!/(k-q)!$, and define $\sum^*$ to be the summation over all pairwise distinct indices. If $\lim_n a_n/b_n=0$, we denote $a_n=o(b_n)$, and if $0<\liminf_na_n/b_n\le\limsup_na_n/b_n<+\infty$, we denote $a_n\asymp b_n$.
Throughout the paper, we use ``$\stackrel{D}{\rightarrow}$" to denote convergence in distribution, ``$\stackrel{P}{\rightarrow}$" for convergence in probability, and ``$\leadsto$" for process convergence in some suitable function space. We use $\ell_{\infty}\left([0,1]^{3}\right)$ to denote the set of bounded functions on $[0,1]^3$.

\subsection{U-statistic and Self-normalization}
In this subsection, we shall develop our test statistics for one change point alternative, i.e.,
\[\mathcal{H}_{1} : \mu_{1}=\cdots=\mu_{k_{1}} \neq \mu_{k_{1}+1}= \cdots=\mu_{n},
\]
where $k_1$ is unknown. In  \cite{hdcp}, a U-statistic based approach was developed and their test targets at the dense alternative since the power is a monotone function of $\sqrt{n}\|\Delta\|_2/\|\Sigma\|_F^{1/2}$, where $\Delta$ is the difference between pre-break and post-break means, i.e., $\Delta=\mu_{n}-\mu_{1}$, and 
$\Sigma$ is the covariance matrix of $X_i$. Thus their test may not be powerful if the change in mean is sparse and $\|\Delta\|_2$ is small. Note that several tests have been developed to capture sparse alternatives as mentioned in Section~\ref{sec:intro}. In practice, when there is no prior knowledge of the alternative for a given data set at hand, it would be helpful to have a test that can be adaptive to different types and 
magnitudes of the change. To this end, we shall adopt the $L_q$ norm-based approach, as initiated by \cite{xu2016} and \cite{he2018}, and develop a class of test statistics indexed by $q\in2\N$, and then combine these tests to achieve the adaptivity.

Denote $X_i=(X_{i,1},\ldots,X_{i,p})^T$. For any positive even number $q\in2\N$, consider the following two-sample U-statistic of order $(q,q)$,
$$
T_{n, q}(k)=\frac{1}{P_{q}^{k} P_{q}^{n-k}} \sum_{l=1}^{p} \sum^*_{1 \leq i_{1},\ldots,i_{q} \leq k}\sum^*_{k+1 \leq j_{1},\ldots,j_{q} \leq n}\left(X_{i_{1}, l}-X_{j_{1}, l}\right) \cdots\left(X_{i_{q}, l}-X_{j_{q}, l}\right),
$$
for any $k=q, \cdots, n-q$.
Simple calculation shows that $\mathbb{E}\left[T_{n,q}(k)\right]=0$ for any $k=q,\cdots,n-q$ under the null hypothesis, and $\mathbb{E}\left[T_{n,q}(k_1)\right]=\|\Delta\|_q^q$ under the alternative.  When $q\in 2\N+1$ (i.e., $q$ is odd) and under the alternative, $\mathbb{E}\left[T_{n,q}(k_1)\right]=\sum_{j=1}^{p}\delta_j^q\not=\|\Delta\|_q^q$ where $\Delta=(\delta_1,\cdots,\delta_p)^T$. This is the main reason we focus on the statistics corresponding to even $q$s since for an odd $q$, $\sum_{j=1}^{p}\delta_j^q=0$ does not imply $\Delta=0$.

If the change point location $k_1=\lf \tau_1n\rf$, $\tau_1\in (0,1)$ is known, then we would use $T_{n,q}(k_1)$ as our test statistic. As implied by the asymptotic results shown later, we have that under the null,
\[\left(\frac{\tau_1(1-\tau_1)}{n\|\Sigma\|_q}\right)^{q/2}\frac{T_{n,q}(k_1)}{\sqrt{q!}}\stackrel{D}{\rightarrow} N(0,1),\]
under suitable moment and weak dependence assumptions on the components of $X_t$. In practice, a typical approach is to replace $\|\Sigma\|_q$ by a ratio-consistent estimator, which is available for $q=2$ [see \cite{chenqin2010}], but not for general $q\in 2\N$. In practice, the location $k_1$ is unknown, which adds additional complexity to the variance estimation and motivates \cite{hdcp} to use the idea of self-normalization [\cite{shao2010a}, \cite{shao2010}] in the case $q=2$. Self-normalization is a nascent inferential method [\cite{lobato2001}, \cite{shao2010a}] that has been developed for low and fixed-dimensional parameter in a low dimensional time series. It uses an inconsistent variance estimator to yield an asymptotically pivotal statistic, and does not involve any tuning parameter or involves less number of tuning parameters compared to traditional procedures. See \cite{shao2015} 
for a comprehensive review of recent developments for low dimensional time series. There have been two recent extensions to the high-dimensional setting: \cite{wangshao2019} adopted  a one sample U-statistic with trimming and extended self-normalization to inference for the mean of high-dimensional time series; \cite{hdcp} used a two sample U-statistic and extended the self-normalization (SN)-based change point test in \cite{shao2010} to high-dimensional independent data. Both papers are $L_2$ norm based, and this seems to be the first time that a $L_q$-norm based approach is extended to high-dimensional setting via self-normalization.

Following \cite{hdcp}, we consider the following self-normalization procedure. Define 
\begin{align*}
U_{n,q}(k;s,m)&=\sum_{l=1}^{p} \sum^*_{s \leq i_{1},\ldots,i_{q} \leq k}\sum^*_{k+1 \leq j_{1},\ldots,j_{q} \leq m}\left(X_{i_{1}, l}-X_{j_{1}, l}\right) \cdots\left(X_{i_{q}, l}-X_{j_{q}, l}\right),
\end{align*}
which is an un-normalized version of $T_{n,q}$ applied to the subsample $(X_s,\cdots,X_m)$. Let  
$$
W_{n,q}(k;s,m) :=\frac{1}{m-s+1} \sum_{t=s+q-1}^{k-q} U_{n,q}(t;s, k)^{2}+\frac{1}{m-s+1} \sum_{t=k+q}^{m-q} U_{n,q}(t; k+1,m)^{2},
$$
The self-normalized statistic is given by
$$
\widetilde{T}_{n,q} :=\max _{k=2q, \ldots, n-2q} \frac{U_{n,q}(k ; 1, n)^{2}}{W_{n,q}(k ; 1, n)}.
$$

\begin{remark}
If we want to test for multiple change points, we can use the scanning idea presented in \cite{zhang2018} and  \cite{hdcp} and construct the following statistic:
$$
T_{n,q}^{*}:=\max _{2q\le l_1\le l_2-2q} \frac{U_{n,q}\left(l_{1} ; 1, l_{2}\right)^{2}}{W_{n,q}\left(l_{1} ; 1, l_{2}\right)}+\max _{m_1+2q-1\le m_2\le n-2q} \frac{U_{n,q}\left(m_{2} ; m_1, n\right)^{2}}{W_{n,q}\left(m_{2} ; m_1,n\right)}.
$$
We shall skip further details as the asymptotic theory and computational implementation are fairly straightforward.
\end{remark}


\subsection{Limiting Null Distribution}

Before presenting our main theorem, we need to make the following assumptions.
\begin{assumption}

Suppose $Z_1,\ldots,Z_n$ are i.i.d. copies of $Z_0$ with mean 0 and covariance matrix $\Sigma$, and the following conditions hold.

\begin{enumerate}
    \item There exists $c_{0}>0$ not depending $n$ such that inf $_{i=1, \ldots, p_{n}} \operatorname{Var}\left(Z_{0, i}\right) \geq c_{0}$.
    \item $Z_0$ has up to $8$-th moments, with $\sup_{1\le j \le p}\E[Z_{0,j}^8]\le C,$ and for $h=2, \ldots, 8$ there exist constants $C_{h}$ depending on $h$ only and a constant $r>2$ such that $$\left|\operatorname{cum}\left(Z_{0, l_{1}}, \ldots, Z_{0, l_{h}}\right)\right| \leq C_{h}\left(1 \vee \max _{1 \leq i, j \leq h}\left|l_{i}-l_{j}\right|\right)^{-r}.$$
\end{enumerate}
\label{cumr}
\end{assumption}

\begin{remark}[Discussion of Assumptions]
The above cumulant assumption is implied by geometric moment contraction [cf. Proposition 2 of \cite{wu2004}] or physical dependence measure proposed by \cite{wu2005} [cf. Section 4 of \cite{shaowu2007}], or $\alpha$-mixing [\cite{andrews1991}, \cite{zhur1975}] in the time series setting. It basically imposes weak dependence among the $p$ components in the data. Our theory holds as long as a permutation of $p$ components satisfies the cumulant assumption, since our test is invariant to the permutation within the components. 
\label{assdis}
\end{remark}

To derive the limiting null distribution for $\widetilde{T}_{n,q}$, we need to define some useful intermediate processes. Define
\begin{align*}
D_{n,q}(r;[a,b])&=U_{n,q}(\lf nr\rf;\lf na\rf+1,\lf nb\rf)\\
&=\sum_{l=1}^{p}\sum^*_{\lf na\rf+1 \leq i_{1},\ldots,i_{q} \leq \lf nr \rf}\sum^*_{\lf nr \rf+1 \leq j_{1},\ldots,j_{q} \leq \lf nb \rf}\left(X_{i_{1}, l}-X_{j_{1}, l}\right) \cdots\left(X_{i_{q}, l}-X_{j_{q}, l}\right),
\end{align*}
for any $0\le a<r<b\le 1$. Note that under the null, $X_i$'s have the same mean. Therefore, we can rewrite $D_{n,q}$ as
\begin{align*}
D_{n,q}(r;[a,b])=&\sum_{l=1}^{p} \sum^*_{\lf na \rf+1 \leq i_{1},\ldots,i_{q} \leq \lf nr\rf}\sum^*_{\lf nr\rf+1 \leq j_{1},\ldots,j_{q} \leq \lf bn\rf}\left(X_{i_{1}, l}-X_{j_{1}, l}\right) \cdots\left(X_{i_{q}, l}-X_{j_{q}, l}\right)\\
=&\sum_{l=1}^{p} \sum^*_{\lf na \rf+1 \leq i_{1},\ldots,i_{q} \leq \lf nr\rf}\sum^*_{\lf nr\rf+1 \leq j_{1},\ldots,j_{q} \leq \lf bn\rf}\left(Z_{i_{1}, l}-Z_{j_{1}, l}\right) \cdots\left(Z_{i_{q}, l}-Z_{j_{q}, l}\right)\\
=&\sum_{c=0}^q(-1)^{q-c}\binom{q}{c}P^{\lf nr\rf-\lf na\rf-c}_{q-c}P^{\lf nb\rf-\lf nr\rf-q+c}_cS_{n,q,c}(r;[a,b]).
\end{align*}
In the above expression, considering the summand for each $c = 0,1,\ldots,q,$ we can define, for any $0\le a<r<b\le 1$,
$$
S_{n,q,c}(r;[a,b])=\sum_{l=1}^{p}\sum^*_{\lf na\rf+1\leq i_1,\cdots,i_c\leq\lf nr\rf}\sum^*_{\lf nr\rf+1\leq j_1,\cdots,j_{q-c}\le \lf nb\rf}\left(\prod_{t=1}^{c} Z_{i_{t}, l} \prod_{s=1}^{q-c} Z_{j_{s}, l}\right),
$$
if $\lfloor n r\rfloor \geq\lfloor n a\rfloor+ 1$ and $\lfloor n b\rfloor \geq\lfloor n r\rfloor+ 1,$ and 0 otherwise.

\begin{theorem}
If Assumption \ref{cumr} holds, then under the null and for a finite set $I$ of positive even numbers, we have that 
$$\Big\{a_{n,q}^{-1} S_{n, q, c}(\cdot ;[\cdot, \cdot])\Big\}_{q\in I,0\le c\le q} \leadsto \Big\{Q_{q, c}(\cdot ;[\cdot, \cdot])\Big\}_{q\in I,0\le c\le q}$$
in $\ell_{\infty}\left([0,1]^{3}\right)$ jointly over $q\in I, 0\le c\le q$,  where $a_{n,q}=\sqrt{n^q\sum^p_{l_1,l_2=1}\Sigma^q_{l_1,l_2}}=\sqrt{n^q\|\Sigma\|_q^q}$, and $Q_{q, c}$ are centered Gaussian processes. Furthermore, the covariance of $Q_{q, c_1}$ and $Q_{q, c_2}$ is given by 
$$ 
\operatorname{cov}\left(Q_{q, c_1}(r_1;[a_1, b_1]),Q_{q, c_2}(r_2;[a_2, b_2])\right)=\binom{C}{c}c!(q-c)!(r-A)^c(R-r)^{C-c}(b-R)^{q-C},
$$
where $(r,R)=(\min\{r_1,r_2\},\max\{r_1,r_2\})$,$(a,A)=(\min\{a_1,a_2\},\max\{a_1,a_2\})$,$(b,B)=(\min\{b_1,b_2\},\\ \max\{b_1,b_2\}),$ and $(c,C)=(\min\{c_1,c_2\},\max\{c_1,c_2\})$. 
Additionally, $Q_{q_1,c_1}$ and $Q_{q_2,c_2}$ are mutually independent if $q_{1} \neq q_{2}\in2\N$.
\label{Q}
\end{theorem}
For illustration, consider the case when $a_1<a_2<r_1<r_2<b_1<b_2$ and $c_1\le c_2$. We have
$$ 
\operatorname{cov}\left(Q_{q, c_1}(r_1;[a_1, b_1]),Q_{q, c_2}(r_2;[a_2, b_2])\right)=\binom{c_2}{c_1}c_1!(q-c_1)!(r_1-a_2)^{c_1}(r_2-r_1)^{c_2-c_1}(b_1-r_2)^{q-c_2},
$$
which implies, for example,
$$
\operatorname{var}\left[Q_{q, c}(r;[a,b])\right]=c!(q-c)!(r-a)^{c}(b-r)^{q-c}.
$$
The proof of Theorem \ref{Q} is long and is deferred to the supplement.


\begin{theorem}
Suppose Assumption \ref{cumr} holds.  Then for  a finite set $I$ of positive even numbers,
$$\Big\{n^{-q}a_{n,q}^{-1} D_{n,q}(\cdot ;[\cdot, \cdot])\Big\}_{q\in I} \leadsto \Big\{G_{q}(\cdot ;[\cdot, \cdot])\Big\}_{q\in I}$$
in $\ell_{\infty}\left([0,1]^{3}\right)$ jointly over $q\in I$, where
$$
G_{q}=\sum_{c=0}^{q}(-1)^{q-c} \binom{q}{c}(r-a)^{q-c}(b-r)^{c}Q_{q,c}
$$
and $Q_{q,c}$ is given in Theorem \ref{Q}.  Furthermore, for $q_1\not=q_2\in 2\N$, $G_{q_1}$ and $G_{q_2}$ are independent.
Consequently, we have that under the null, 
$$
\widetilde{T}_{n,q} \stackrel{\mathcal{D}}{\longrightarrow} \widetilde{T}_q=\sup _{r \in[0,1]} \frac{G_q(r ; 0,1)^{2}}{\int_{0}^{r} G_q(u ; 0, r)^{2} d u+\int_{r}^{1} G_q(u ; r, 1)^{2} d u}.
$$

\label{convT}
\end{theorem}

It can be derived that the $G_q(\cdot;[\cdot,\cdot])$ is a Gaussian process with the following covariance structure: 
$$
\operatorname{var}[G_q(r;[a,b])]=\sum^q_{c=0}\binom{q}{c}^2c!(q-c)!(r-a)^{2q-c}(b-r)^{q+c}=q!(r-a)^q(b-r)^q(b-a)^q.
$$
When $r_1=r_2=r$,
\begin{align*}
\operatorname{cov}(G_q(r;[a_1,b_1]),G_q(r;[a_2,b_2]))=q!(r-A)^q(b-r)^q(B-a)^q.  
\end{align*}
When $r_1\not=r_2$,
$$\operatorname{cov}(G_q(r_1;[a_1,b_1]),G_q(r_2;[a_2,b_2]))=q![(r-A)(b-R)(B-a)-(A-a)(R-r)(B-b)]^q,$$
where $(r,R,a,A,b,B,c,C)$ is defined in Theorem~\ref{Q}.
The limiting null distribution $\widetilde{T}_q$ is pivotal and its critical values can be simulated as done in \cite{hdcp} for the case $q=2$. The simulated critical values and their corresponding realizations for $q=2,4,6$  are available upon request. For a practical reason, we did not pursue the larger $q$, such as $q=8,10$, since larger $q$ corresponds to more trimming on the two ends and the finite sample performance when $q=6$ is already very promising for detecting sparse alternatives, see Section \ref{sec:simulation}. An additional difficulty with larger $q$ is the associated computation cost and complexity in its implementation. 

\begin{remark}
Compared to \cite{he2018}, we  assume the $8$th moment conditions, which is weaker than the uniform sub-Gaussian type conditions in their condition A.4(2), although the latter condition seems to be exclusively used for deriving the limit of the test statistic corresponding to $q=\infty$. Furthermore,  since their strong mixing condition with exponential decay rate [cf. condition A.4(3) of \cite{he2018}] implies our cumulant assumption \ref{cumr} [see \cite{andrews1991}, \cite{zhur1975}], our overall assumption is weaker than condition A.4 in \cite{he2018}. Despite the weaker assumptions, our results are stronger, as we derived the asymptotic independence of 
several stochastic process indexed by $q\in 2\N$, which implies the asymptotic independence of U-statistics indexed by $q\in 2\N$.

Note that our current formulation does not include the $q=\infty$ case, which corresponds to $L_{\infty}$ norm of mean difference $\|\Delta\|_{\infty}$.  The $L_{\infty}$-norm based test was developed by  \cite{yu20} and  their test statistic is based on CUSUM statistics 
\[Z_n(s)=\sqrt{\frac{s(n-s)}{n}}\left(\frac{1}{s}\sum_{i=1}^{s}X_i-\frac{1}{n-s}\sum_{i=s+1}^{n}X_i\right)\]
and takes the form $T_n=\max_{s_0\le s\le n-s_0}\|Z_n(s)\|_{\infty}$, where $s_0$ is the boundary removal parameter. They did not obtain the asymptotic distribution of $T_n$ but showed that a bootstrap CUSUM test statistic  is able to approximate the finite sample distribution of $T_n$ using a modification of Gaussian and bootstrap approximation techniques developed by \cite{cck13,cck17}. Given the asymptotic independence between $L_q$-norm based U statistic and $L_{\infty}$-norm based test statistic [\cite{he2018}] in the two-sample testing text, we would conjecture that $T_n$ test statistic in \cite{yu20} is asymptotically independent of our $\widetilde{T}_{n,q}$ for any $q\in 2\N$ under suitable moment and weak componentwise dependence conditions. A rigorous investigation is left for future work.


\end{remark}

\subsection{Adaptive Test}
Let $I$ be a set of $q\in 2\N$ (e.g. \{2,6\}).  Since $\widetilde{T}_{n,q}$s are asymptotically independent for different $q\in I$ under the null, we can combine their corresponding $p$-values and form an adaptive test. For example, we may use $p_{ada}=\min_{q\in I}p_q$, where $p_q$ is the $p$-value corresponding to $\widetilde{T}_{n,q}$,  
as a new statistic. Its $p$-value is equal to $1-(1-p_{ada})^{|I|}$. Suppose we want to perform a  level-$\alpha$ test, it is equivalent to conduct tests based on $\widetilde{T}_{n,q},\forall q\in I$ at level $1-(1-\alpha)^{1/|I|}$, and reject the null if one of the statistics exceeds its critical value. Therefore, we only need to compare each $\widetilde{T}_{n,q}$ with its $(1-\alpha)^{1/|I|}$-quantile of the corresponding limiting null distribution. 

As we explained before, a smaller $q$ (say $q=2$) tends to have higher power under the dense alternative, which is also the main motivation for the proposed method in \cite{hdcp}. On the contrary, a larger $q$ has a higher power under the sparse alternative, as $\lim_{q\rightarrow\infty}\|\Delta_n\|_q=\|\Delta_n\|_{\infty}$. 
Therefore, with the adaptive test, we can achieve high power under both dense and sparse alternatives with asymptotic size still equal to $\alpha$. This adaptivity will be confirmed by our asymptotic power analysis presented in Section~\ref{sec:power}  and simulation results  presented in Section~\ref{sec:simulation}.

\subsection{Power Analysis}
\label{sec:power}

\begin{theorem}
Assume that the change point location is at $k_1=\lf \tau_1n\rf$ with the change in the mean equal to $\Delta_n=(\delta_{n,1},\ldots,\delta_{n,p})^T$. Suppose Assumption \ref{cumr}, and the following conditions on $\Delta_{n}$ hold. We have

\begin{enumerate}
    \item If $n^{q/2}\left\|\Delta_{n}\right\|_{q}^q /\|\Sigma\|_{q}^{q / 2} \rightarrow \infty, \text { then } \widetilde{T}_{n,q} \stackrel{\mathcal{P}}{\longrightarrow}\infty$;
    \item If $n^{q/2}\left\|\Delta_{n}\right\|_{q}^q /\|\Sigma\|_{q}^{q / 2} \rightarrow 0, \text { then } \widetilde{T}_{n,q} \stackrel{\mathcal{D}}{\longrightarrow} \widetilde{T}_q$;
    \item If $n^{q/2}\left\|\Delta_{n}\right\|_{q}^q /\|\Sigma\|_{q}^{q / 2} \rightarrow \gamma\in(0,+\infty)$, then 
    $$
    \widetilde{T}_{n,q} \stackrel{\mathcal{D}}{\longrightarrow} \sup _{r \in[0,1]} \frac{\{G_q(r ; 0,1)+\gamma J_q(r, 0,1)\}^{2}}{\int_{0}^{r}\{G_q(u ; 0, r)+\gamma J_q(u, 0, r)\}^{2} d u+\int_{r}^{1}\{ G_q(u ; r, 1)+\gamma J_q(u, r, 1)\}^{2} du},
    $$
    where
    $$
   J_q(r, a, b) :=\left\{\begin{array}{lr}{\left(\tau_1-a\right)^{q}(b-r)^{q}} & {a<\tau_1 \leq r<b} \\ {(r-a)^{q}\left(b-\tau_1\right)^{q}} & {a<r<\tau_1<b} \\ {0} & {\tau_1<a \text { or } \tau_1>b}\end{array}\right..$$
\end{enumerate}
\label{power}
\end{theorem}

\begin{remark}
The following example illustrates the power behavior using different $q\in 2\N$. For simplicity, we assume $\Sigma=I_p$ and consider a change in the mean equal to $\Delta_n=\delta\cdot(\bm{1}_d,\bm{0}_{p-d})^T$. In addition to demonstrating that large (small) $q$ is favorable to  the sparse (dense) alternatives, our local asymptotic power results stated in Theorem \ref{power} also allow us to provide a rule to classify an alternative, which is given by
$$
\left\{\begin{array}{lr}{sparse} & {d=o(\sqrt{p})}   \\{in~between} &{d\asymp\sqrt{p}} \\ {dense}&  {\sqrt{p}=o(d)} \end{array}\right..
$$
To have a nontrivial power, it suffices to have  $n^{q/2}\left\|\Delta_{n}\right\|_{q}^q /\|\Sigma\|_{q}^{q / 2}=dn^{q/2}\delta^q/\sqrt{p}=\gamma\in(0,+\infty)$, which implies $\delta\asymp (\sqrt{p}/d)^{1/q}n^{-1/2}$. Therefore, when $d=o(\sqrt{p})$, a smaller $\delta$ corresponds to a larger $q$. On the contrary, when $\sqrt{p}=o(d)$, a smaller $q$ that yields a larger $\delta$ is preferable to have higher power. Similar argument still holds for more general $\Delta_n$ and $\Sigma$, as long as we have a similar order for $\|\Delta_n\|_q^q$ and $\|\Sigma\|_q^q$, and the latter one is guaranteed by Assumption \ref{cumr}. 

We can summarize the asymptotic powers of the tests under different alternatives in the following table.
 Note that when at least one single-$q$ based test obtains asymptotically nontrivial power (power 1), our adaptive test can also achieve nontrivial power (power 1). 
\begin{table}[H]
\centering
\begin{tabular}{|c|c|c|c|c|}
\hline
Alternative & $\delta$ & $I=\{2\}$ & $I=\{q\}$ & $I=\{2,q\}$ \\ \hline
\multirow{6}{*}{\shortstack{Dense \\$\sqrt{p}=o(d)$}} & $\delta=o(p^{1/4}d^{-1/2}n^{-1/2})$ & $\alpha$ & $\alpha$ & $\alpha$  \\ \cline{2-5} 
 & $\delta\asymp p^{1/4}d^{-1/2}n^{-1/2}$ & $\beta_1\in(\alpha,1)$ & $\alpha$ &  $(\alpha,\beta_1)$ \\ \cline{2-5} 
 & {$p^{1/4}d^{-1/2}n^{-1/2}=o(\delta),$ } & \multirow{2}{*}{1} & \multirow{2}{*}{$\alpha$} & \multirow{2}{*}{1} \\ 
 & \& $\delta=o(p^{1/2q}d^{-1/q}n^{-1/2})$ &  & & \\ \cline{2-5}
 & $\delta\asymp p^{1/2q}d^{-1/q}n^{-1/2}$ & 1 & $(\alpha,1)$ & 1 \\ \cline{2-5}
 & $p^{1/2q}d^{-1/q}n^{-1/2}=o(\delta)$ & 1 & 1 & 1 \\ \hline
\multirow{6}{*}{\shortstack{Sparse \\$d=o(\sqrt{p}$)}} 
& $\delta=o(p^{1/2q}d^{-1/q}n^{-1/2})$  & $\alpha$ & $\alpha$ & $\alpha$  \\ \cline{2-5} 
& $\delta\asymp p^{1/2q}d^{-1/q}n^{-1/2}$  & $\alpha$ & $\beta_2\in(\alpha,1)$ & $(\alpha,\beta_2)$  \\ \cline{2-5} 
 & {$p^{1/2q}d^{-1/q}n^{-1/2}=o(\delta),$ } & \multirow{2}{*}{$\alpha$}  & \multirow{2}{*} {1}& \multirow{2}{*}{1} \\ 
 & \& $\delta=o(p^{1/4}d^{-1/2}n^{-1/2})$ &  &  & \\ \cline{2-5}
 & $\delta\asymp p^{1/4}d^{-1/2}n^{-1/2}$ & $(\alpha,1)$ & 1 & 1 \\ \cline{2-5}
 & $p^{1/4}d^{-1/2}n^{-1/2}=o(\delta)$ & 1 & 1 & 1 \\ \hline
\end{tabular}
\caption{Asymptotic powers of single-$q$ and adaptive tests}
\end{table}

\end{remark}

\cite{lgs20} recently studied the detection of a sparse change in the high-dimensional mean vector under the Gaussian assumption as a minimax testing problem. Let $\rho^2=\min(k_1,n-k_1)\|\Delta_n\|_2^2$. In the fully dense case, i.e., when $\|\Delta_n\|_0=p$, where $\|\Delta_n\|_0$ denotes the $L_0$ norm, Theorem 8 in \cite{lgs20} stated that the minimax rate is given by $\rho^2\asymp \|\Sigma\|_F\sqrt{\log\log(8n)}\vee \|\Sigma\|_{s}\log\log(8n)$. Thus under the assumption that $k_1/n=\tau_1\in (0,1)$, the $L_2$-norm based test in \cite{hdcp} achieves the rate optimality up to a logarithm factor. Consequently, any adaptive test based on $I$ is rate optimal (up to a logarithm factor) as long as $2\in I$. 

In the special case $\Sigma=I_p$, the minimax rate is given by
$$
\rho^{2} \asymp\left\{\begin{array}{lr}
\sqrt{p \log \log (8 n)} & \text { if } d \geq \sqrt{p \log \log (8 n)} \\
d\log\left(\frac{e p \log \log (8 n)}{d^{2}}\right) \vee \log \log (8 n) & \text { if } d<\sqrt{p \log \log (8 n)}
\end{array}\right..
$$
Recall that $\Delta_n=\delta\cdot(\bm{1}_d,\bm{0}_{p-d})^T$. 
In the sparse setting $d=o(\sqrt{p})$ and under the assumptions that $d\asymp p^{-v}$,$v\in (0,1/2)$ and $d>\log\log(8n)$,  the minimax rate is $d$ (up to a logarithm factor), which corresponds to $\delta\asymp n^{-1/2}$. Our $L_q$-norm based test is not minimax rate optimal since the detection boundary is $(\sqrt{p}/d)^{1/q} n^{-1/2}$, which gets closer to $n^{-1/2}$ as $q\in 2\N$ gets larger. 
In the dense setting $\sqrt{p}=o(d)$ and under the assumptions that $d\asymp p^{-v}$,$v\in (1/2,1)$ and $d>\sqrt{p\log\log(8n)}$, the minimax rate is $\sqrt{p}$ (up to a logarithm factor), which corresponds to $\delta\asymp p^{1/4}/\sqrt{nd}$. Therefore the $L_2$-norm based test in \cite{hdcp} is again rate optimal (up to a logarithm factor).


\section{Change-point Estimation}
\label{sec:estimation}

In this section, we investigate the change-point location estimation based on change-point test statistics we proposed in Section~\ref{sec:test}. Specifically, Section~\ref{subsec:singleestimate} presents convergence rate for the argmax of SN-based test statistic  upon suitable standardization. Section~\ref{subsec:WBS} proposes a combination of wild binary segmentation (WBS, \cite{fry2014}) algorithm with our SN-based test statistics for both single-$q$ test and adaptive test to estimate multiple change points. 

\subsection{Single Change-point Estimation}
\label{subsec:singleestimate}

In this subsection, we propose to estimate  the location of a change point assuming that the data is generated from the following single change-point model,
\[X_t=\mu_1+\Delta_n {\bf 1}(t>k^*)+Z_t,~t=1,\cdots,n,\]
where $k^*=k_1=\lfloor \tau^* n\rfloor $ is the location of change point. In the literature, it is common to focus on the convergence rate of the estimators of the relative location $\tau^* \in (0,1)$, that is, we shall focus on the convergence rate of $\hat{\tau}=\hat{k}/n$, where $\hat{k}$ is an estimator for $k^*$.


Given the discussions about size and power properties of the SN-based test statistic in Section~\ref{sec:test}, it is natural to use the argmax of the test statistic as the estimator for $k^*$.
That is, we define 

$$\hat{k}=\operatorname{argmax}_{k=2q, \ldots, n-2q} \frac{U_{n,q}(k ; 1, n)^{2}}{W_{n,q}(k ; 1, n)}.$$ 
To present the convergence rate for $\hat{\tau}$, we shall introduce the following assumptions. 

\begin{assumption}\label{ass}

\begin{enumerate}
    \item $tr(\Sigma^4) = o(\|\Sigma\|_F^4)$;
    \item $\sum_{l_1,...,l_h = 1}^p cum(Z_{0,l_1},...,Z_{0,l_h})^2 \leq  C\|\Sigma\|_F^h$, for $h = 2,...,6$;
    \item $\|\Sigma\|_F = o(n\|\Delta_n\|_2^2)$.
\end{enumerate}
\end{assumption}


Let $\gamma_{n,q} = n^{q/2}\|\Delta_n\|_q^q/\|\Sigma\|_q^{q/2}$ so $\gamma_{n,2}= n\|\Delta_n\|^2/\|\Sigma\|_F$. We have the following convergence rate of $\hat \tau$ for the case $q=2$.

\begin{theorem}\label{thm:consistency}
Suppose Assumption \ref{ass} holds and $q=2$. It holds that $\hat{\tau} - \tau^*=o_p( \gamma_{n,2}^{-1/4+\kappa})$ as $n \wedge p \rightarrow \infty$, for any $0 < \kappa < 1/4$. 
\end{theorem}

\begin{remark}
Assumption \ref{ass} (1) and (2) have been assumed in \cite{hdcp}, and they are implied by Assumption \ref{cumr}; see Remark 3.2 in \cite{hdcp}. Assumption \ref{ass}(3) is equivalent to  $\gamma_{n,2} \rightarrow \infty$, which implies that $\hat{\tau}$ is a consistent estimator of $\tau^*$. Note that even in the low-dimensional setting, no convergence rate for the argmax of SN-based statistic (standarized by the sample size) is obtained in \cite{shao2010}. Thus this is the first time the asymptotic rate for the argmax of a SN-based test statistic is studied. On the other hand, the proof for the more general case $q\in 2\N$ is considerably more involved than the special case $q=2$ and is deferred to future investigation.


\end{remark}

\subsection{Multiple Change-point Estimation}
\label{subsec:WBS}

In practice, the interest is often in the change point estimation or segmentation, when the presence of change points is confirmed by testing or based on prior knowledge. In the high-dimensional context, the literature on change point estimation is relatively scarce; see \cite{cho2016}, \cite{wang2018} and \cite{hdcp}. Here we shall follow the latter two papers and use the wild  binary segmentation [\cite{fry2014}] coupled with our test developed for a single $q$ or adaptive test to estimate the number and location of change points. Note that the standard binary segmentation procedure may fail when the change in means is not monotonic, as shown in \cite{hdcp}
via simulations.

For any integers $s,e$ satisfying $2q\leq s+2q-1\leq e-2q\leq n-2q$, define 
$$
Q_{n,q}(s,e):=\max_{b=s+2q-1,...,e-2q}\frac{U_{n,q}^2(b;s,e)}{W_{n,q}(b;s,e)},
$$

Note that $Q_{n,q}(s,e)$ is essentially the statistic $T_{n,q}$ based on the sub-sample $(X_s,...,X_e)$. Denote a random sample of $(s_m,e_m)$ s.t. $2q\leq s_m+2q-1\leq e_m-2q\leq n-2q$ as $F_n^M$, where the sample is drawn independently with replacement of size $M$. In practice, we may require the segments to be slightly longer to reduce unnecessary fluctuations of the critical values.
Then define $\hat{\xi}_{n,M,q}=\max_{m=1,\cdots,M} Q_{n,q}(s_m,e_m)$ and we stop the algorithm if $\hat{\xi}_{n,M,q}\le \xi_{n,q}$, where $\xi_{n,q}$ is some threshold to be specified below, and estimate the change point otherwise; see Algorithm~\ref{WBS} for details. 

One anonymous reviewer asked whether it is possible to derive the limiting distribution of $\hat{\xi}_{n,M,q}$ under the null, which turns out to be challenging for two reasons: (1) The SN-based test statistic for different intervals could be highly dependent, especially when the two intervals overlap by a lot; (2) the number of such randomly generated intervals is usual large, and it would be more valuable to develop an asymptotic distribution under the assumption that both sample size and number of intervals go to infinity. It seems difficult to use the classical argument for this problem, and we shall leave this for future investigation.

To obtain the threshold value $\xi_{n,q}$ as needed in the Algorithm 1, we generate $R$ standard Gaussian samples each of which has sample size $n$ and  dimension $p$.
For the $r$-th sample ($r=1,\ldots,R)$, we calculate
$$
\hat{\xi}^{(r)}_{n,M,q}=\max_{m=1,\cdots,M} Q_{n,q}^{(r)}(s_m,e_m),
$$
where $Q_{n,q}^{(r)}(s_m,e_m)$ is the SN-based test statistic applied to the $r$th Gaussian simulated sample. We can take $\xi_{n,q}$ to be the 95\% quantile of $\{\hat{\xi}^{(r)}_{n,M,q}\}_{r=1}^{R}$. Since the self-normalized test statistic is asymptotically pivotal, the above threshold  $\xi_{n,q}$  is expected to approximate the 95\% quantile of the finite sample distribution of maximized SN-based test statistic applied to $M$ randomly drawn sub-samples from the original data. 


\begin{algorithm} [H]
    \begin{algorithmic}[1] 
    \Function{WBS}{$S,E$}
    \If{$E-S<4q-1$}
    \State STOP
    \Else 
    \State $\mathcal{M}_{s,e}\leftarrow$ set of those $1\leq m\leq M$ s.t. $S\leq s_m,e_m\leq E,e_m-s_m\geq 4q-1$
    \State $m_{q}\leftarrow$argmax$_{m\in \mathcal{M}_{s,e}}Q_{n,q}(s_m,e_m)$
    \If{$Q_{n,q}(s_{m_q},e_{m_q})>\xi_{n,q}$}
    \State add $b_0\leftarrow$argmax$_bU_{n,q}(b;s_{m_q},e_{m_q})/W_{n,q}(b;s_{m_q},e_{m_q})$ to set of estimated CP
    \State WBS$(S,b_0)$
    \State WBS$(b_0+1,E)$
    \Else
    \State STOP
    \EndIf
    \EndIf
    \EndFunction
    \end{algorithmic}  
\caption{WBS Algorithm for a given $q\in 2\N$}
\label{WBS}
\end{algorithm}  

To apply the adaptive test, we calculate $\hat{\xi}_{n,M,q}^{(r)}$ with $r$-th sample using different $q\in I$. Denote $q_I:=\max_{q\in I} q$ We calculate $p$-value for each single-$q$ based statistic and select the most significant one for location estimation, which gives the adaptive version; see Algorithm \ref{WBSada}. 

\begin{algorithm} [H]
    \begin{algorithmic}[1]  
    \Function{WBS}{$S,E$}
    \If{$E-S<4q_I-1$}
    \State STOP
    \Else 
    \State $p_0=0.05$
    \For{$q$ in $I$}
    \State $\mathcal{M}_{s,e}\leftarrow$ set of those $1\leq m\leq M$ s.t. $S\leq s_m,e_m\leq E,e_m-s_m\geq 4q_I-1$
    \State $m_q\leftarrow$argmax$_{m\in \mathcal{M}_{s,e}}Q_{n,q}(s_m,e_m)$
    \State $p_q=R^{-1}\#\Big\{Q_{n,q}(s_{m_q},e_{m_q})>\xi^{(r)}_{n,M,q}\Big\}_{r=1}^R$
    \If{$p_q<p_0$ for current $q$}
    \State $b_0\leftarrow$argmax$_bU_{n,q}(b;s_{m_q},e_{m_q})/W_{n,q}(b;s_{m_q},e_{m_q})$
    \State $p_0\leftarrow p_q$
    \State NEXT
    \EndIf
    \EndFor
    \State add $b_0$ to set of estimated CP
    \State WBS$(S,b_0)$
    \State WBS$(b_0+1,E)$
    \EndIf
    \EndFunction
    \end{algorithmic}  
\caption{Adaptive WBS Algorithm}
\label{WBSada}
\end{algorithm}  

\section{Numerical Studies}
\label{sec:simulation}

In this section, we present numerical results to examine the finite sample performance of our testing and estimation method in comparison with the existing alternatives. Section~\ref{sec:sim1} shows the size and power for the single change point tests; Section~\ref{sec:sim1.5} presents the estimation result when there is one single change-point; Section~\ref{sec:sim2} compares several WBS-based estimation methods for multiple change point estimation, including the INSPECT method in \cite{wang2018}. Finally, we apply our method to a real data set in Section~\ref{sec:data}.

\subsection{Single Change Point Testing}
\label{sec:sim1}


In this subsection, we examine the size and power property of our single-$q$ and adaptive tests in comparison with the one in \cite{en2013} (denoted as EH), which seems to be the only adaptive 
method in the literature. The data  $X_i\sim N(\mu_i,\Sigma)$, where $\mu_i=0$ for $i=1,\cdots,n$ under the null. We set $(n,p)=(200, 100)$ and $(400,200)$ and performed 2000 Monte carlo replications. 
We consider four different configurations of $\Sigma=(\sigma_{ij})$ as follows,
$$
\sigma_{ij}=\left\{\begin{array}{lr}{\mathds{1}_{i=j}} & {\text {Id}} \\{0.5^{|i-j|}} & {\text {AR(0.5)}} \\{0.8^{|i-j|}} & {\text {AR(0.8)}} \\{\mathds{1}_{i=j}+0.25\cdot\mathds{1}_{i\not=j}} & {\text {CS}} \\\end{array}\right..
$$
They correspond to independent components (Id), auto-regressive model with order $1$ (AR(0.5) and AR(0.8)) and compound symmetric (CS), respectively. The first three configurations imply weak dependence among components so satisfy Assumption \ref{cumr}, whereas the compound symmetric covariance matrix corresponds to 
strong dependence among components and violates our assumption. The size of our tests, including $\widetilde{T}_{n,q}$ at a single $q=2,4,6$ and combined tests with $\mathcal{I}=(2,4)$, $(2,6)$ and $(2,4,6)$ are presented  in Table \ref{sizeonent}. It appears that all tests are oversized when $\Sigma$ is compound symmetric, which is somewhat expected since the strong dependence among components brings non-negligible errors in asymptotic approximation. As a matter of fact, we conjecture that our limiting null distribution $\widetilde{T}_q$ no longer holds in this case.
Below we shall focus our comments on the first three configurations (Id, AR(0.5) and AR(0.8)).

The size for $q=2$ (i.e., the test in \cite{hdcp}) appears quite accurate except for some degree of  under-rejection in the Id case. For $q=4$, it is oversized and its size seems inferior to the case  $q=6$, which also shows some over-rejection for the AR(1) models when $(n,p)=(200,100)$, but the size distortion improves quite a bit when we increase $(n,p)$ to $(400, 200)$. Among the three combined tests, there are apparent over-rejections for $\mathcal{I}=(2,4)$, $(2,4,6)$ for the AR(1) models, and the test corresponding to $\mathcal{I}=(2,6)$ exhibits the most accurate size overall. By contrast, the EH shows serious size distortions in all settings with some serious over-rejection when the componentwise dependence is strong (e.g., AR(0.8) and CS), which is consistent with the fact that its validity strongly relies on the Gaussian and componentwise independence assumptions. We also checked the sensitivity of the size with respect to nonGaussian assumptions and observe serious distortion for EH when the data is generated from a nonGaussian distribution (results not shown). Overall,  the adaptive test with $\mathcal{I}=(2,6)$ seems preferred to all other tests (including the adaptive test with $\mathcal{I}=(2,4,6)$) in terms of size accuracy.

 


\centerline{Please insert Table \ref{sizeonent} here!}

To investigate the power, we let $\mu_i=0$ for $i\le n/2$ and $\mu_i=\sqrt{\delta/d}\cdot(\bm{1}_{d},\bm{0}_{p-d})^T$ for $i>n/2$. We take $\delta=1,2$ and $d=3$, which corresponds to a sparse alternative; and let $d=p$ to examine the power under the dense alternative; see Table~\ref{poweronent}. In the case of sparse alternative, we can see that the powers corresponding to $q=4$ and $q=6$ are much higher than that for $q=2$, which is consistent with our intuition. When $q=4$, the power is slightly higher than that for $q=6$, which might be explained by the over-rejection with $q=4$ (in the case of AR(1) models),  and we expect no power gain as we increase $q$ to $8,10$ etc, so the results for these larger $q$ are not included. Also for larger $q$, there is more trimming involved as the maximum runs from $2q$ to $n-2q$ in our test statistics, so if the change point occurs outside of the range $[2q, n-2q]$, our test has little power. In the dense alternative case, the power for $q=2$ is the highest as expected, and the  power for $q=4$ is again slightly higher than that for $q=6$.

The power of the combined tests (i.e., ${\mathcal I}=(2,4)$ or $(2,6)$ or $(2,4,6)$) is always fairly close to the best single one within the set. For example, the power for $(2,6)$ is very close to the power for $q=6$ in the sparse case and is quite close to the power for $q=2$ in the dense case, indicating the adaptiveness of the combined test. In the sparse case, the powers for ${\mathcal I}=(2,4)$ and  $(2,4,6)$ are slightly higher than that for $(2,6)$, which could be related to the over-rejection of  the tests with ${\mathcal I}=(2,4)$ and  $(2,4,6)$, especially when the data is generated from AR(1) models. Overall, the adaptive tests (i.e., (2,4), (2,6) or (2,4,6)) have a good all-around power behavior against both sparse and dense alternatives and are preferred choices when there is no prior knowledge about the type of alternative the data falls into. Since the size for $(2,6)$ is more accurate than that for (2,4) and (2,4,6), we slightly favor the $(2,6)$ combination.
EH exhibits high power for all settings, but it is at the cost of serious size distortion. We shall not present size-adjusted power as the serious distortion is too great to recommend its use when there are componentwise dependence in the data. 


\centerline {Please insert Tables \ref{poweronent} here!}

\subsection{Estimation for Single Change-point}
\label{sec:sim1.5}

In this subsection, we present the square root of mean-square-error (RMSE, multiplied by 1000 for readability) of SN-based location estimators and compare with the EH-based estimator under the same settings as we used in Section~\ref{sec:sim1}.

For both dense and sparse alternatives, the proposed estimators (i.e., SN(2), SN(4) and SN(6)) perform better than the EH method when the signal is relatively weak (i.e., $\delta=1,2$). However, as the signal becomes stronger (i.e., $\delta=4$), the  EH method
can outperform ours in the identity covariance matrix case. On the other other hand, the performance of the EH estimator apparently deterioates as the cross-sectional dependence gets stronger, indicating its strong reliance on the componentwise independence assumption. It is interesting to note that the SN-based method performs fairly well, even in the case of compound symmetric covariane matrix, and SN(6) outperforms the other two in all settings. A theoretical justification for the latter phenomenon would be intriguing. 

\centerline {Please insert Tables \ref{loconent} here!}

\subsection{Estimation for Multiple Change Points}
\label{sec:sim2}

In the following simulations, we compare our WBS-based method with the INSPECT method proposed by \cite{wang2018}. Following \cite{hdcp}, we generate 100 samples of i.i.d. standard normal variables $\{Z_t\}_{t=1}^n$ with $n=120,p=50$. The 3 change points are located at $30, 60$ and $90$. Denote the changes in mean by $\bm{\theta_1},\bm{\theta_2},\bm{\theta_3},$  with $\bm{\theta_1}=-\bm{\theta_2}=2\sqrt{k_1/d_1}\cdot(\bm{1}_{d_1},\bm{0}_{p-d_1}),\bm{\theta_3}=2\sqrt{k_2/d_2}\cdot(\bm{1}_{d_2},\bm{0}_{p-d_2})$. We use, e.g., Dense(2.5) to denote dense changes with $d_i=p=50,k_i=2.5$ for $i=1,2,3$ and Sparse(4) to denote sparse changes with $d_i=5,k_i=4$ for $i=1,2,3$. In particular, Dense($2.5$) \& Sparse($4$) refers to $k_1=2.5,k_2=4,d_1=5,d_2=50$, where we have a mixture of dense and sparse changes. 

We compare WBS with INSPECT, for which we use default parameters with the "InspectChangepoint" package in R. We use 2 different metrics for evaluating the performance of different methods. One is to calculate the mean square errors (MSE) of the estimated number of change points. The other metric takes the accuracy of location estimation into account. We utilize the correlated rand index (CRI), which can measure the accuracy of change point location estimation. 
See \cite{rand1971}, \cite{hub1985} and \cite{hdcp} for more details. For perfect estimation, the calculated CRI is 1. In general it is a number between 0 and 1 and the more precise we estimate the change point locations, the higher CRI we get. We average the CRI for all Monte Carlo replications and record the average rand index (ARI). We report the MSE and ARI of different methods based on 100 replications in Table \ref{simwbs}.

When there are only sparse changes and $\delta=2.5$, the performance of adaptive procedure (WBS(2,6)) is similar to WBS(6), whose estimation is much more accurate than WBS(2) and INSPECT. When we strengthen the signal by increasing $\delta$ from $2.5$ to $4$, the detection power of all methods increase, but instead INSPECT has the best estimation accuracy in this case, closely followed by WBS(6) and WBS(2,6). In the case of purely dense changes with $\delta=2.5$, the performance of WBS(2) dominates the others and WBS(2,6) is the second best. When we increase $\delta$ from $2.5$ to $4$ in this setting, the adaptive test slightly outperforms INSPECT, and its performance is comparable to WBS(2). For both dense settings, the performance of WBS(6) is rather poor. We can see that under all these four settings, the performance of WBS(2,6) is always close to the best, indicating its adaptiveness to different types of change points. Moreover, when there is a mixture of dense and sparse changes, the adaptive method outperforms all the others. In practice, the type of changes is often unknown, and therefore our adaptive procedure could be appealing for practitioners.

\centerline{Please insert Table \ref{simwbs} here!}

\subsection{Real data illustration}
\label{sec:data}
In this subsection, we study the genomic micro-array data set that contains log intensity ratios of 43 individuals with bladder tumor, measured at 2215 different loci. The data was available in R package \texttt{ecp} and was also studied by \cite{hdcp} and \cite{wang2018}. We compare our results with theirs.

We take the first 200 loci for our study. For the WBS algorithm, we generate 10000 samples from i.i.d. standard normal distributions with $(n,p)=(200,43)$, and draw  5000 random intervals  to calculate the supremum statistics and get the 98\%-quantile as our critical value. The change points detected by WBS at 0.98 level and the 20 most significant points detected by INSPECT are given as follows.
$$
\begin{array}{ll} 
q=2 & 33,39,46,74,97,102,135,155,173,191\\  
q=6 & 15,32,44,59,74,91,116,134,158,173,186\\  
q=2,6 & 15,32,38,44,59,74,91,97,102,116,134,158,173,186,191\\
\text{INSPECT} & 15,26,28,33,36,40,56,73,91,97,102,119,131,134,135,146,155,174,180,191
\end{array}
$$
We can see that the set of change points detected by the adaptive WBS method is roughly a union of the sets corresponding to two single WBS methods (that is, for a single $q$), which suggests that the adaptive WBS method captures both sparse and dense alternatives as expected. In particular, 32(33),44(46),74,134(135), 158(155),173 are detected by both single methods, 38(39),97,102,191 are detected only by $q=2$, and 15, 59, 91,116,186 only by $q=6$. The set of the change points detected by adaptive WBS method overlaps with the set for INSPECT by a lot, including 15, 32(33), 38(36), 74(73), 91, 97, 102, 116(119), 134, 158(155), 173(174), and 191. 
It is worth noting that the change points at locations 91, 97, 191 were only detected by one of two single WBS methods and INSPECT, whereas the adaptive WBS method is able to capture with its good all-round power property again a broad range of alternatives.  
In Figure \ref{acgh}, we plot the log intensity ratios of the first 10 individuals at first 200 loci, and the locations of the change points estimated by the adaptive method. 

\centerline{Please insert Figure \ref{acgh} here!}

This example clearly demonstrates the usefulness of the proposed adaptive test and corresponding WBS-based estimation method. An important practical choice is the threshold, which can be viewed as a tuning parameter in the implementation of WBS algorithm. We shall leave its choice for future investigation.

\section{Conclusion}
In this paper, we propose a class of asymptotically pivotal statistics for testing a mean change in high-dimensional independent data. The test statistics are formed on the basis of an unbiased estimator of $q$-th power of the $L_q$ norm of the mean change via U-statistic and self-normalization. They are asymptotically independent for different $q\in 2\N$, and therefore, we can form an adaptive test by taking the minimum of $p$-values corresponding to test statistics indexed by $q\in 2\N$. The resulting test is shown to have good overall power against both dense and sparse alternatives via theory and simulations. On the estimation front, we obtain the convergence rate for the argmax of SN-based test statistic  standardized by sample size under the one change-point model and $q=2$.  We also combine our tests with WBS algorithm to estimate multiple change points. As demonstrated by our simulations, the WBS-based estimation method inherits the advantage of the adaptive test, as it outperforms other methods under the setting where there is a mixture of dense and sparse change points, and has close-to-best performance for purely dense and purely sparse cases.

To conclude, we mention that it would be interesting  to extend our adaptive test to the high-dimensional time series setting, for which a trimming parameter seems necessary to accommodate weak temporal dependence in view of recent work by \cite{wangshao2019}. 
In addition, the focus of this paper is on mean change, whereas in practice the interest could be on other high-dimensional parameters, such as vector of marginal quantiles, variance-covariance matrix, and even high-dimensional distributions. It remains to be seen whether some extensions to these more general parameters are possible in the high-dimensional environment.  We shall leave these open problems for future research.
\bigskip

\newpage
\begin{table}[H]
\centering
\begin{tabular}{|c|c|c|c|c|c|c|c|c|}
\hline
\multirow{2}{*}{DGP} & \multirow{2}{*}{$(n,p)$} & \multicolumn{7}{c|}{$\alpha=$5\%} \\ \cline{3-9} 
 &  & $q=2$ & $q=4$ & $q=6$ & $q=2,4$ & $q=2,6$ & $q=2,4,6$ & EH\\ \hline
\multirow{2}{*}{Id} & (200,100) & 0.028 & 0.065 & 0.056 & 0.052 & 0.032 & 0.045 & 0.01\\ \cline{2-9} 
 & (400,200) & 0.036 & 0.068 & 0.051 & 0.055 & 0.041 & 0.045 & 0 \\ \hline
 \multirow{2}{*}{AR(0.5)} & (200,100) & 0.049 & 0.109 & 0.077 & 0.087 & 0.063 & 0.085  & 0.111 \\ \cline{2-9} 
 & (400,200) & 0.043 & 0.089 & 0.058 & 0.081 & 0.056 & 0.074  & 0.093

 \\ \hline
  \multirow{2}{*}{AR(0.8)} & (200,100) & 0.051 & 0.12 & 0.079 & 0.097 & 0.063 & 0.086 & 0.613 \\ \cline{2-9} 
 & (400,200) & 0.045 & 0.094 & 0.046 & 0.082 & 0.047 & 0.069 & 0.66\\ \hline
  \multirow{2}{*}{CS} & (200,100) & 0.095 & 0.103 & 0.081 & 0.109 & 0.088 & 0.098 & 0.729\\ \cline{2-9} 
 & (400,200) & 0.11 & 0.085 & 0.061 & 0.116 & 0.099 & 0.104 & 0.898\\ \hline
\end{tabular}
\caption{Size for one change point test}
\label{sizeonent}
\end{table}

\begin{table}[H]
\centering
\begin{tabular}{|c|c|c|c|c|c|c|c|c|c|c|}
\hline
\multirow{2}{*}{$\mathcal{H}_1$} & \multirow{2}{*}{DGP} & \multirow{2}{*}{$\delta$} &
\multirow{2}{*}{$(n,p)$} & \multicolumn{7}{c|}{$\alpha=$5\%} \\ \cline{5-11} 
 & &  &  & $q=2$ & $q=4$ & $q=6$ & $q=2,4$ & $q=2,6$ & $q=2,4,6$ & EH \\ \hline
\multirow{16}{*}{Sparse} & \multirow{4}{*}{Id} &\multirow{2}{*} {1} & (200,100)& 0.742 & 0.981 & 0.962 & 0.982 & 0.967 & 0.981 & 0.844 \\ \cline{4-11} 
 & &  & (400,200) & 0.94 & 1 & 1 & 1 & 1 & 1 & 0.998 \\ \cline{3-11} 
 & & \multirow{2}{*}{2} & (200,100) & 0.995 & 1 & 1 & 1 & 1 & 1 & 1 \\ \cline{4-11} 
 & &  & (400,200) & 1 & 1 & 1 & 1 & 1 & 1 & 1  \\  \cline{2-11}

 & \multirow{4}{*}{AR(0.5)} & \multirow{2}{*}{1} & (200,100) & 0.566  & 0.947 & 0.91 & 0.933 & 0.894 & 0.93 &  0.809\\ \cline{4-11} 
 & & & (400,200) & 0.82 & 1 & 1 & 1 & 0.996 & 1 & 0.994 \\ \cline{3-11}
 & & \multirow{2}{*}{2} & (200,100) & 0.94 & 1 & 1 & 1 & 0.999 & 1 & 0.999 \\ \cline{4-11} 
 & & & (400,200) & 0.998  & 1 & 1 & 1 & 1 & 1 & 1  \\ \cline{2-11}
 
 & \multirow{4}{*}{AR(0.8)} & \multirow{2}{*}{1} & (200,100)  & 0.298 & 0.887 & 0.84 & 0.883 & 0.82 & 0.876 & 0.912  \\ \cline{4-11} 
 & &  & (400,200) & 0.522 & 0.99 & 0.994 & 0.99 & 0.988 & 0.992 & 1 \\ \cline{3-11} 
 & &\multirow{2}{*}{2} & (200,100) & 0.703 & 0.997 & 0.995 & 0.997 & 0.994 & 0.997 & 0.994 \\ \cline{4-11} 
 & & & (400,200) & 0.928 & 1 & 1 & 1 & 1 & 1 & 1 \\ \cline{2-11}
  & \multirow{4}{*}{CS} & \multirow{2}{*}{1} & (200,100)  & 0.231 & 0.971 & 0.937 & 0.966 & 0.927 & 0.962 & 0.964 \\ \cline{4-11} 
 & & & (400,200) & 0.226 & 1 & 1 & 1 & 1 & 1 & 1 \\ \cline{3-11} 
 & & \multirow{2}{*}{2} & (200,100) & 0.592 & 1 & 1 & 1 & 1 & 1 & 1 \\ \cline{4-11} 
 & & & (400,200) & 0.656 & 1 & 1 & 1 & 1 & 1 & 1 \\ \hline

\multirow{16}{*}{Dense} & \multirow{4}{*}{Id} & \multirow{2}{*}{1} & (200,100)& 0.718 & 0.326 & 0.292 & 0.677 & 0.661  & 0.645 & 0.444 \\ \cline{4-11} 
 & & & (400,200) & 0.94 & 0.348 & 0.282 & 0.912 & 0.896 & 0.89 & 0.82 \\ \cline{3-11} 
 & & \multirow{2}{*}{2} & (200,100) & 0.995 & 0.567 & 0.612 & 0.991 & 0.99 & 0.985 & 0.978\\ \cline{4-11} 
 & & & (400,200) & 1 & 0.682 & 0.628 & 1 & 1 & 1 & 1 \\ \cline{2-11}

 & \multirow{4}{*}{AR(0.5)} & \multirow{2}{*}{1} & (200,100)& 0.589 & 0.357 & 0.312 & 0.581 & 0.554 & 0.552 & 0.566\\ \cline{4-11} 
 & & & (400,200) & 0.808 & 0.36 & 0.338 & 0.762 & 0.754 & 0.732 & 0.832 \\ \cline{3-11} 
 & & \multirow{2}{*}{2} & (200,100) & 0.927 & 0.616 & 0.573 & 0.917 & 0.909 & 0.899 & 0.93\\ \cline{4-11} 
 & & & (400,200) & 0.994 & 0.68 & 0.608 & 0.99 & 0.988 & 0.984 & 0.998 \\  \cline{2-11}
  
  & \multirow{4}{*}{AR(0.8)} & \multirow{2}{*}{1} & (200,100)  & 0.385 & 0.33 & 0.262 & 0.41 & 0.358 & 0.376 & 0.831  \\ \cline{4-11} 
 & & & (400,200) & 0.502 & 0.32 & 0.242 & 0.474 & 0.436 & 0.422 & 0.912 \\ \cline{3-11} 
 & & \multirow{2}{*}{2} & (200,100) & 0.693 & 0.537 & 0.474 & 0.699 & 0.656 & 0.667 & 0.94 \\ \cline{4-11} 
 & & & (400,200) & 0.872 & 0.564 & 0.51 & 0.866 & 0.848 & 0.84 & 0.986 \\  \cline{2-11}
    
    & \multirow{4}{*}{CS} & \multirow{2}{*}{1} & (200,100)  & 0.345 & 0.277 & 0.235 & 0.363 & 0.33 & 0.338 & 0.84 \\ \cline{4-11} 
 & & & (400,200) & 0.36 & 0.284 & 0.214 & 0.368 & 0.334 & 0.348 & 1 \\ \cline{3-11} 
 & & \multirow{2}{*}{2} & (200,100) & 0.544 & 0.455 & 0.414 & 0.551 & 0.526 & 0.532 & 0.919 \\ \cline{4-11} 
 & & & (400,200) & 0.588 & 0.474 & 0.424 & 0.584 & 0.55 & 0.562  & 1 \\ \hline
\end{tabular}
\caption{Power for one change point test under different alternatives}
\label{poweronent}
\end{table}

\begin{table}[H]
\centering
\begin{tabular}{|c|c|c|c|c|c|c|c|c|c|}
\hline
\multirow{2}{*}{$\delta$} & \multirow{2}{*}{Method} & \multicolumn{4}{c|}{Sparse} & \multicolumn{4}{c|}{Dense} \\ \cline{3-10} 
                          &       & Id    & AR(0.5) & AR(0.8) & CS    & Id    & AR(0.5) & AR(0.8) & CS    \\ \hline
\multirow{4}{*}{1}  & SN(2) & 38.7  & 53.7    & 72.3    & 84.5  & 41.1  & 50.6    & 72.6    & 91.4  \\ \cline{2-10} 
                          & SN(4) & 20.3  & 24.5    & 26.7    & 20.9  & 44.6  & 43.0    & 49.1    & 49.6  \\ \cline{2-10} 
                          & SN(6) & 18.5  & 22.0    & 22.3    & 19.6  & 33.1  & 31.6    & 32.6    & 31.4  \\ \cline{2-10} 
                          & EH    & 150.3 & 214.9   & 300.6   & 326.8 & 155.6 & 216.9   & 291.3   & 332.3 \\ \hline
\multirow{4}{*}{2}  & SN(2) & 26.0  & 33.5    & 41.7    & 44.3  & 27.5  & 36.6    & 54.8    & 76.8  \\ \cline{2-10} 
                          & SN(4) & 14.4  & 17.5    & 19.6    & 16.3  & 37.9  & 38.3    & 45.8    & 48.8  \\ \cline{2-10} 
                          & SN(6) & 12.1  & 14.1    & 14.9    & 11.9  & 29.7  & 28.6    & 31.9    & 30.5  \\ \cline{2-10} 
                          & EH    & 41.4  & 90.2    & 196.8   & 272.7 & 40.1  & 110.8   & 210.2   & 286.6 \\ \hline
\multirow{4}{*}{4}  & SN(2) & 21.8  & 24.8    & 29.5    & 30.4  & 20.7  & 26.1    & 39.2    & 64.2  \\ \cline{2-10} 
                          & SN(4) & 12.1  & 14.7    & 16.5    & 14.1  & 22.8  & 27.8    & 39.0    & 44.3  \\ \cline{2-10} 
                          & SN(6) & 9.9   & 10.8    & 11.6    & 10.1  & 26.1  & 26.5    & 29.1    & 33.6  \\ \cline{2-10} 
                          & EH    & 8.7   & 16.7    & 66.8    & 153.2 & 9.7   & 29.9    & 109.4   & 209.6 \\ \hline

\end{tabular}
\caption{RMSE (multiplied by $10^3$) for one change point location estimation under different alternatives}
\label{loconent}
\end{table}

\begin{figure}[H]
\centering
\includegraphics[width=17cm]{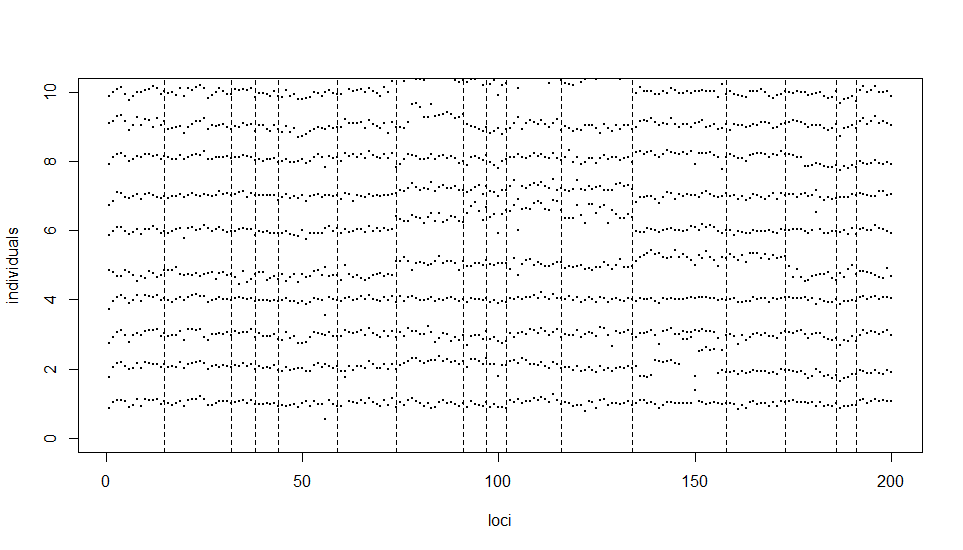}
\caption{ACGH data of the first 10 individuals at first 200 loci. The dashed lines represent the locations of the change points detected.}
\label{acgh}
\end{figure}

\begin{table}[H]
    \small
    \centering
    \begin{tabular}{c|c|ccccccc|c|c}
    \hline
    \multirow{2}{*}{}  &  \multirow{2}{*}{} & \multicolumn{7}{c|}{$\hat{N}-N$} & \multirow{2}{*}{MSE} & \multirow{2}{*}{ARI}  \\
    \cline{3-9}
     & & -3 & -2 & -1 & 0 & 1 & 2 & 3 & & \\
     \hline
     \multirow{5}{*}{Sparse(2.5)} & WBS-SN(2) & 0 & 1 & 11 & 75 & 13 & 0 & 0 & 0.28 & 0.8667 \\
     \cline{2-11}
      & WBS-SN(4)  & 0 & 0 & 0 & 98 & 2 & 0 & 0 & 0.02 & 0.958 \\
        \cline{2-11}
      & WBS-SN(6)  & 0 & 0 & 0 & 94 & 5 & 1 & 0 & 0.09 & 0.9552 \\
      \cline{2-11}
      & WBS-SN(2,6)  & 0 & 0 & 0 & 90 & 10 & 0 & 0 & 0.1 & 0.9489 \\
     \cline{2-11}
      & INSPECT  & 0 & 26 & 0 & 69 & 5 & 0 & 0 & 1.09 & 0.7951 \\
     \hline
     \multirow{5}{*}{Sparse(4)} & WBS-SN(2) & 0 & 0 & 0 & 86 & 14 & 0 & 0 & 0.14 & 0.9188 \\
     \cline{2-11}
      & WBS-SN(4)  & 0 & 0 & 0 & 98 & 2 & 0 & 0 & 0.02 & 0.9684 \\
        \cline{2-11}
      & WBS-SN(6)  & 0 & 0 & 0 & 94 & 5 & 1 & 0 & 0.09 & 0.9707 \\
      \cline{2-11}
      & WBS-SN(2,6)  & 0 & 0 & 0 & 90 & 10 & 0 & 0 & 0.1 & 0.9678 \\
     \cline{2-11}
      & INSPECT  & 0 & 0 & 0 & 91 & 8 & 1 & 0 & 0.12 & 0.9766 \\
     \hline
          \multirow{5}{*}{Dense(2.5)} & WBS-SN(2) & 0 & 2 & 10 & 74 & 13 & 1 & 0 & 0.35 & 0.8662 \\
     \cline{2-11}
      & WBS-SN(4)  & 94 & 4 & 2 & 0 & 0 & 0 & 0 & 8.64 & 0.0263 \\
        \cline{2-11}
      & WBS-SN(6)  & 70 & 20 & 7 & 3 & 0 & 0 & 0 & 7.17 & 0.1229 \\
      \cline{2-11}
      & WBS-SN(2,6)  & 5 & 5 & 7 & 64 & 9 & 0 & 0 & 0.91 & 0.7809 \\
     \cline{2-11}
      & INSPECT  & 0 & 40 & 0 & 46 & 13 & 0 & 1 & 1.82 & 0.6656 \\
     \hline
     \multirow{5}{*}{Dense(4)} & WBS-SN(2) & 0 & 0 & 0 & 85 & 13 & 2 & 0 & 0.21 & 0.9186 \\
     \cline{2-11}
      & WBS-SN(4) & 47 & 33 & 14 & 6 & 0 & 0 & 0 & 5.69 & 0.2748 \\
        \cline{2-11}
      & WBS-SN(6)  & 46 & 28 & 21 & 5 & 0 & 0 & 0 & 5.47 & 0.2642 \\
      \cline{2-11}
      & WBS-SN(2,6)  & 0 & 0 & 0 & 87 & 13 & 0 & 0 & 0.13 & 0.9214 \\
     \cline{2-11}
      & INSPECT  & 0 & 7 & 0 & 68 & 22 & 2 & 1 & 0.67 & 0.9027 \\
     \hline
       & WBS-SN(2) & 0 & 1 & 12 & 73 & 14 & 0 & 0 & 0.3 & 0.8742 \\
     \cline{2-11}
      Sparse(2.5) & WBS-SN(4) & 0 & 0 & 62 & 37 & 1 & 0 & 0 & 0.63 & 0.7855 \\
        \cline{2-11}
      \& & WBS-SN(6)  & 0 & 0 & 60 & 38 & 2 & 0 & 0 & 0.62 & 0.7743 \\
      \cline{2-11}
     Dense(4) & WBS-SN(2,6)  & 0 & 0 & 0 & 91 & 9 & 0 & 0 & 0.09 & 0.9439 \\
     \cline{2-11}
      & INSPECT  & 0 & 21 & 1 & 70 & 6 & 1 & 1 & 1.04 & 0.8198 \\
     \hline
    \end{tabular}
    \caption{Multiple change-point  estimation}
    \label{simwbs}
\end{table}

\newpage

\bibliographystyle{asa}
\bibliography{ref}

\begin{thebibliography}{41}
\newcommand{\enquote}[1]{``#1''}
\expandafter\ifx\csname natexlab\endcsname\relax\def\natexlab#1{#1}\fi

\bibitem[{Aminikhanghahi and Cook(2017)}]{ami2017}
Aminikhanghahi, S. and Cook, D.~J. (2017), \enquote{A survey of methods for
  time series change point detection,} \textit{Knowledge and Information
  Systems}, 51, 339--367.

\bibitem[{Andrews(1991)}]{andrews1991}
Andrews, D. (1991), \enquote{Heteroskedasticity and autocorrelation consistent
  covariant matrix estimation,} \textit{Econometrica}, 59, 817--858.

\bibitem[{Aue and Horv{\'a}th(2013)}]{au2013}
Aue, A. and Horv{\'a}th, L. (2013), \enquote{Structural breaks in time series,}
  \textit{Journal of Time Series Analysis}, 34, 1--16.

\bibitem[{Chan et~al.(2013)Chan, Horv{\'a}th, and Hu{\v{s}}kov{\'a}}]{chan2013}
Chan, J., Horv{\'a}th, L., and Hu{\v{s}}kov{\'a}, M. (2013),
  \enquote{Darling--Erd{\H{o}}s limit results for change-point detection in
  panel data,} \textit{Journal of Statistical Planning and Inference}, 143,
  955--970.

\bibitem[{Chen and Zhang(2015)}]{chen2015graph}
Chen, H. and Zhang, N. (2015), \enquote{Graph-based change-point detection,}
  \textit{The Annals of Statistics}, 43, 139--176.

\bibitem[{Chen and Gupta(2011)}]{chen2011}
Chen, J. and Gupta, A.~K. (2011), \textit{Parametric Statistical Change Point
  Analysis: with Applications to Genetics, Medicine, and Finance}, Springer
  Science \& Business Media.

\bibitem[{Chen and Qin(2010)}]{chenqin2010}
Chen, S.~X. and Qin, Y.-L. (2010), \enquote{A two-sample test for
  high-dimensional data with applications to gene-set testing,} \textit{The
  Annals of Statistics}, 38, 808--835.

\bibitem[{Chernozhukov et~al.(2013)Chernozhukov, Chetverikov, and Kato}]{cck13}
Chernozhukov, V., Chetverikov, D., and Kato, K. (2013), \enquote{Gaussian
  approximations and multiplier bootstrap for maxima of sums of
  high-dimensional random vectors,} \textit{Annals of Statistics}, 41,
  2786--2819.

\bibitem[{Chernozhukov et~al.(2017)Chernozhukov, Chetverikov, and Kato}]{cck17}
--- (2017), \enquote{Central limit theorems and bootstrap in high dimensions,}
  \textit{Annals of Probability}, 45, 2309--2352.

\bibitem[{Cho(2016)}]{cho2016}
Cho, H. (2016), \enquote{Change-point detection in panel data via double CUSUM
  statistic,} \textit{Electronic Journal of Statistics}, 10, 2000--2038.

\bibitem[{Cs\"org\"o and Horv{\' a}th(1997)}]{csorgo1997}
Cs\"org\"o, M. and Horv{\' a}th, L. (1997), \textit{Limit Theorems in
  Change-Point Analysis. Wiley Series in Probability and Statistics.}, Wiley.

\bibitem[{Enikeeva and Harchaoui(2019)}]{en2013}
Enikeeva, F. and Harchaoui, Z. (2019), \enquote{High-dimensional change-point
  detection with sparse alternatives,} \textit{The Annals of Statistics}, 47,
  2051--2079.

\bibitem[{Fryzlewicz(2014)}]{fry2014}
Fryzlewicz, P. (2014), \enquote{Wild binary segmentation for multiple
  change-point detection,} \textit{The Annals of Statistics}, 42, 2243--2281.

\bibitem[{He et~al.(2020)He, Xu, Wu, and Pan}]{he2018}
He, Y., Xu, G., Wu, C., and Pan, W. (2020), \enquote{Asymptotically independent
  U-Statistics in high-dimensional testing,} \textit{The Annals of Statistics},
  forthcoming.

\bibitem[{Horv{\'a}th and Hu{\v{s}}kov{\'a}(2012)}]{ho2012}
Horv{\'a}th, L. and Hu{\v{s}}kov{\'a}, M. (2012), \enquote{Change-point
  detection in panel data,} \textit{Journal of Time Series Analysis}, 33,
  631--648.

\bibitem[{Hubert and Arabie(1985)}]{hub1985}
Hubert, L. and Arabie, P. (1985), \enquote{Comparing partitions,}
  \textit{Journal of Classification}, 2, 193--218.

\bibitem[{Jeng et~al.(2010)Jeng, Cai, and Li}]{jeng10}
Jeng, X., Cai, T., and Li, H. (2010), \enquote{Optimal sparse segment
  identification with application in copy number variation analysis,}
  \textit{Journal of the American Statistical Association}, 105, 1156--1166.

\bibitem[{Jirak(2015)}]{ji2015}
Jirak, M. (2015), \enquote{Uniform change point tests in high dimension,}
  \textit{The Annals of Statistics}, 43, 2451--2483.

\bibitem[{Kley et~al.(2016)Kley, Volgushev, Dette, and Hallin}]{kley}
Kley, T., Volgushev, S., Dette, H., and Hallin, M. (2016), \enquote{Quantile
  spectral processes: asymptotic analysis and inference,} \textit{Bernoulli},
  22, 1770--1807.

\bibitem[{Liu et~al.(2020)Liu, Gao, and Samworth}]{lgs20}
Liu, H., Gao, C., and Samworth, R. (2020), \enquote{Minimax Rates in Sparse
  High-Dimensional Changepoint Detection,} \textit{Annals of Statistics, to
  appear}.

\bibitem[{Lobato(2001)}]{lobato2001}
Lobato, I.~N. (2001), \enquote{Testing that a dependent process is
  uncorrelated,} \textit{Journal of the American Statistical Association}, 96,
  1066--1076.

\bibitem[{Perron(2006)}]{pe2006}
Perron, P. (2006), \enquote{Dealing with structural breaks,} \textit{Palgrave
  Handbook of Econometrics}, 1, 278--352.

\bibitem[{Rand(1971)}]{rand1971}
Rand, W.~M. (1971), \enquote{Objective criteria for the evaluation of
  clustering methods,} \textit{Journal of the American Statistical
  Association}, 66, 846--850.

\bibitem[{Shao(2010)}]{shao2010a}
Shao, X. (2010), \enquote{A self-normalized approach to confidence interval
  construction in time series,} \textit{Journal of the Royal Statistical
  Society, Series, B.}, 72, 343--366.

\bibitem[{Shao(2015)}]{shao2015}
--- (2015), \enquote{Self-normalization for time series: a review of recent
  developments,} \textit{Journal of the American Statistical Association}, 110,
  1797--1817.

\bibitem[{Shao and Wu(2007)}]{shaowu2007}
Shao, X. and Wu, W.~B. (2007), \enquote{Local whittle estimation of fractional
  integration for nonlinear processes,} \textit{Econometric Theory}, 23,
  899--929.

\bibitem[{Shao and Zhang(2010)}]{shao2010}
Shao, X. and Zhang, X. (2010), \enquote{Testing for change points in time
  series,} \textit{Journal of the American Statistical Association}, 105,
  1228--1240.

\bibitem[{Tartakovsky et~al.(2014)Tartakovsky, Nikiforov, and
  Basseville}]{tar2014}
Tartakovsky, A., Nikiforov, I., and Basseville, M. (2014), \textit{Sequential
  Analysis: Hypothesis Testing and Changepoint Detection}, Chapman and
  Hall/CRC.

\bibitem[{Wang et~al.(2020)Wang, Yu, and Rinaldo}]{wang2018opt}
Wang, D., Yu, Y., and Rinaldo, A. (2020), \enquote{Optimal change point
  detection and localization in sparse dynamic networks,} \textit{forthcoming
  at Annals of Statistics, arXiv preprint arXiv:1809.09602}.

\bibitem[{Wang and Shao(2019)}]{wangshao2019}
Wang, R. and Shao, X. (2019), \enquote{Hypothesis testing for high-dimensional
  time series via self-normalization,} \textit{The Annals of Statistics}, to
  appear.

\bibitem[{Wang et~al.(2019)Wang, Volgushev, and Shao}]{hdcp}
Wang, R., Volgushev, S., and Shao, X. (2019), \enquote{Inference for change
  points in high dimensional data,} \textit{arXiv preprint arXiv:1905.08446}.

\bibitem[{Wang and Samworth(2018)}]{wang2018}
Wang, T. and Samworth, R.~J. (2018), \enquote{High dimensional change point
  estimation via sparse projection,} \textit{Journal of the Royal Statistical
  Society: Series B (Statistical Methodology)}, 80, 57--83.

\bibitem[{Wu(2005)}]{wu2005}
Wu, W.~B. (2005), \enquote{Nonlinear system theory: Another look at
  dependence,} \textit{Proceedings of the National Academy of Sciences USA},
  102, 14150--14154.

\bibitem[{Wu and Shao(2004)}]{wu2004}
Wu, W.~B. and Shao, X. (2004), \enquote{Limit theorems for iterated random
  functions,} \textit{Journal of Applied Probability}, 41, 425--436.

\bibitem[{Xu et~al.(2016)Xu, Lin, Wei, and Pan}]{xu2016}
Xu, G., Lin, L., Wei, P., and Pan, W. (2016), \enquote{An adaptive two-sample
  test for high-dimensional means,} \textit{Biometrika}, 103, 609--624.

\bibitem[{Yu and Chen(2020)}]{yu20}
Yu, M. and Chen, X. (2020), \enquote{Finite sample change point inference and
  identification for high-dimensional mean vectors,} \textit{Journal of Royal
  Statistical Society, Series B, to appear}.

\bibitem[{Zhang and Siegmund(2012)}]{zhang12}
Zhang, N.~R. and Siegmund, D.~O. (2012), \enquote{Model selection for
  high-dimensional multi-sequence change-point problems,} \textit{Statistica
  Sinica}, 22, 1507--1538.

\bibitem[{Zhang et~al.(2010)Zhang, Siegmund, Ji, and Li}]{zhang10}
Zhang, N.~R., Siegmund, D.~O., Ji, H., and Li, J.~Z. (2010), \enquote{Detecting
  simultaneous changepoints in multiple sequences,} \textit{Biometrika}, 97,
  631--645.

\bibitem[{Zhang and Lavitas(2018)}]{zhang2018}
Zhang, T. and Lavitas, L. (2018), \enquote{Unsupervised self-normalized
  change-point testing for time series,} \textit{Journal of the American
  Statistical Association}, 113, 637--648.

\bibitem[{Zhao et~al.(2019)Zhao, Chen, and Lin}]{zhao2019}
Zhao, Z., Chen, L., and Lin, L. (2019), \enquote{Change-point detection in
  dynamic networks via graphon estimation,} \textit{arXiv preprint
  arXiv:1908.01823}.

\bibitem[{Zhurbenko and Zuev(1975)}]{zhur1975}
Zhurbenko, I. and Zuev, N. (1975), \enquote{On higher spectral densities of
  stationary processes with mixing,} \textit{Ukrainian Mathematical Journal},
  27, 364--373.

\end{thebibliography}

\newpage
{\bf\large Supplement to "Adaptive Inference for Change Points in High-Dimensional Data"}

\vskip 1cm
The supplement contains all the technical proofs in Section~\ref{sec:proofs} and some additional simulation results on network change-point detection in Section~\ref{sec:network}.

\section{Technical Appendix}
\label{sec:proofs}

In the following, we will denote $a_n\lesssim b_n$ and $b_n\succsim a_n$ if $\limsup_n a_n/b_n <\infty$.
\begin{proof}[Proof of Theorem \ref{convT}]

Recall that under the null, as $X_i$'s have the same mean,
\begin{align*}
D_{n,q}(r;[a,b])=\sum_{c=0}^q(-1)^{q-c}\binom{q}{c}P^{\lf nr\rf-\lf na\rf-c}_{q-c}P^{\lf nb\rf-\lf nr\rf-q+c}_cS_{n,q,c}(r;[a,b]).
\end{align*}
Therefore, we can calculate the covariance structure of $G_{q}$ based on that of $Q_{q,c}$ given in Theorem \ref{Q}.
$$
\operatorname{var}[G_q(r;[a,b])]=\sum^q_{c=0}\binom{q}{c}^2c!(q-c)!(r-a)^{2q-c}(b-r)^{q+c}=q!(r-a)^q(b-r)^q(b-a)^q.
$$
When $r_1<r_2$,
\begin{align*}
&\operatorname{cov}(G_q(r_1;[a_1,b_1]),G_q(r_2;[a_2,b_2]))\\
=&\sum_{0\le c_1\le c_2\le q}\Big((-1)^{c_1+c_2}\binom{q}{c_1}\binom{q}{c_2}\binom{C}{c}c!(q-C)!\mathds{1}_{r_1>a_2,r_2<b_1}\\
&\cdot(r_1-a_1)^{q-c_1}(b_1-r_1)^{c_1}(r_2-a_2)^{q-c_2}(b_2-r_2)^{c_2}(r-A)^c(R-r)^{C-c}(b-R)^{q-C}\Big).   
\end{align*}
When $r_1>r_2$,
\begin{align*}
&\operatorname{cov}(G_q(r_1;[a_1,b_1]),G_q(r_2;[a_2,b_2]))\\
=&\sum_{0\le c_2\le c_1\le q}\Big((-1)^{c_1+c_2}\binom{q}{c_1}\binom{q}{c_2}\binom{C}{c}c!(q-C)!\mathds{1}_{r_2>a_1,r_1<b_2}\\
&\cdot(r_1-a_1)^{q-c_1}(b_1-r_1)^{c_1}(r_2-a_2)^{q-c_2}(b_2-r_2)^{c_2}(r-A)^c(R-r)^{C-c}(b-R)^{q-C}\Big).   
\end{align*}
When $r_1=r_2=r$,
\begin{align*}
&\operatorname{cov}(G_q(r;[a_1,b_1]),G_q(r;[a_2,b_2]))\\
=&\sum^q_{c=0}\binom{q}{c}^2c!(q-c)!(r-a_1)^{q-c}(b_1-r)^{c}(r-a_2)^{q-c}(b_2-r)^{c}(r-A)^c(b-r)^{q-c}\\
=&q!(r-A)^q(b-r)^q(B-a)^q.  
\end{align*}
For $r_1\not=r_2$, we have 
$$\operatorname{cov}(G_q(r_1;[a_1,b_1]),G_q(r_2;[a_2,b_2]))=q![(r-A)(b-R)(B-a)-(A-a)(R-r)(B-b)]^q.$$
For $q_1\not=q_2$, since covariance of $Q_{q_1,c_1}$ and $Q_{q_2,c_2}$ is 0, we know the covariance of $G_{q_1}$ and $G_{q_2}$ is also 0, since their arbitrary linear combinations are also Gaussian by previous proofs, they are jointly Gaussian and therefore independence is implied by uncorrelation. The rest follows from an application of the continuous mapping theorem. 
\end{proof}

Note that our Assumption \ref{cumr} is a counterpart to the assumption made by Remark 3.2 in \cite{hdcp}. Their results are derived with some weaker assumption (i.e. Assumption 3.1 therein), whose $L_q$-norm based counterpart for is given as follows.

\begin{assumption}
For any $q\in2\N$, the following statements hold:\\
A.1 $\sum_{l_1,l_2,l_3,l_4=1}^p(\Sigma_{l_1l_2}\Sigma_{l_2l_3}\Sigma_{l_3l_4}\Sigma_{l_4l_1})^{q/2}=o\left(\|\Sigma\|_{q}^{2q}\right)$.\\
A.2 $Z_0$ has up to $8-$th moments and there exists a constant $C$ independent of $n$ such that
$$\sum_{l_1,...,l_h=1}^p|\cum(Z_{0,l_1},...,Z_{0,l_h})|^q\leq C\|\Sigma\|_q^{qh/2},$$
\quad \quad for $h=2,\ldots,8$.
\label{8cum}
\end{assumption}

We claim that Assumption \ref{8cum} is implied by Assumption \ref{cumr}.

\begin{proof}[Proof of the claim]
Define
$$
S_{m, h}\left(l_{1}\right):=\left\{1 \leq l_{2}, \ldots, l_{h} \leq p_{n} : \max _{1 \leq i, j \leq h}\left|l_{i}-l_{j}\right|=m\right\}.
$$ 
By triangular inequality, $|l_1-l_2|+|l_2-l_3|+|l_3-l_4|+|l_4-l_1|\ge2\max _{1 \leq i, j \leq 4}|l_{i}-l_{j}|$, and therefore,
\begin{align*}
\sum_{l_{1}, \cdots, l_{4}=1}^{p}(\Sigma_{l_1l_2}\Sigma_{l_2l_3}\Sigma_{l_3l_4}\Sigma_{l_4l_1})^{q/2}=&\sum_{l_{1}=1}^{p_{n}} \sum_{m=0}^{p_{n}} \sum_{l_{2}, \ldots, l_{4} \in S_{m, 4}\left(l_{1}\right)}(\Sigma_{l_1l_2}\Sigma_{l_2l_3}\Sigma_{l_3l_4}\Sigma_{l_4l_1})^{q/2}\\
\le&\sum_{l_{1}=1}^{p_{n}} \sum_{m=0}^{p_{n}}\left|S_{m, 4}\left(l_{1}\right)\right| C_{2}^{2q}(1 \vee m)^{-q r} \\
\lesssim& p_{n} \sum_{m=0}^{p_{n}}(1 \vee m)^{4-2-q r}.
\end{align*}
On the other hand,
\begin{align*}
\sum_{l_{1}, \cdots, l_{h}=1}^{p} \operatorname{cum}^{q}\left(X_{0, l_{1}, n}, \cdots, X_{0, l_{h}, n}\right)&=\sum_{l_{1}=1}^{p_{n}} \sum_{m=0}^{p_{n}} \sum_{l_{2}, \ldots, l_{h} \in S_{m, h}\left(l_{1}\right)} \operatorname{cum}^{q}\left(X_{0, l_{1}, n}, \cdots, X_{0, l_{h}, n}\right)\\
&\leq \sum_{l_{1}=1}^{p_{n}} \sum_{m=0}^{p_{n}}\left|S_{m, h}\left(l_{1}\right)\right| C_{h}^{q}(1 \vee m)^{-q r} \\ 
&\lesssim p_{n} \sum_{m=0}^{p_{n}}(1 \vee m)^{h-2-q r}.
\end{align*}
RHS has order $O\left(p_{n}^{h-qr}\right)$ if $h-qr-1>0$. Now a simple computation shows that Assumption \ref{8cum} is satisfied if $h-qr<h/2$ for $h=2, \ldots, 8,$ and $q=2,\ldots,$ which is equivalent
to $r>2$.
\end{proof}

We are now ready to introduce the following lemma, which is vital in proving the main result.

\begin{lemma}
Under Assumption 2.1, for any $i_1^{(h)},i_2^{(h)},...,  i_q^{(h)}$ that are all distinct, $h = 1,...,8$, and $c = 1,2,...,q$, 
$$
\left|\sum_{l_1,...,l_8 = 1}^p\delta_{n,l_1}^{q-c}\cdots\delta_{n,l_8}^{q-c}\E[Z_{i_1^{(1)},l_1}\cdots Z_{i_c^{(1)},l_1}\cdots Z_{i_1^{(8)},l_8}\cdots Z_{i_c^{(8)},l_8}]\right| \lesssim \|\Delta_n\|_q^{8(q-c)}\|\Sigma\|_q^{4c}   \quad(1)
$$

In particular, for $c=q$, we have
$$\left|\sum_{l_1,...,l_8 = 1}^p\E[Z_{i_1^{(1)},l_1}\cdots Z_{i_c^{(1)},l_1}\cdots Z_{i_1^{(8)},l_8}\cdots Z_{i_c^{(8)},l_8}]\right| \lesssim \|\Sigma\|_q^{4q}.\quad(2)$$

In addition, for any $c = 1,2,...,q-1$,
$$\left|\sum_{l_1,l_2 = 1}^p\delta_{n,l_1}^{q-c}\delta_{n,l_2}^{q-c}\Sigma_{l_1,l_2}^{c}\right| = o\left(\|\Delta_n\|_q^{2(q-c)}\|\Sigma\|_q^{c}\right).\quad(3)$$
	
\label{sumcum}
\end{lemma}

\begin{proof}[Proof of Lemma \ref{sumcum}]
Applying the generalized H\"older's Inequality, we obtain
\begin{align*}
&\left|\sum_{l_1,...,l_8 = 1}^p\delta_{n,l_1}^{q-c}\cdots\delta_{n,l_h}^{q-c}\E[Z_{i_1^{(1)},l_1}\cdots Z_{i_c^{(1)},l_1}\cdots Z_{i_1^{(8)},l_8}\cdots Z_{i_c^{(8)},l_8}]\right|=\left|\E\left(\prod_{u = 1}^8\left[\sum_{l_u = 1}^p\delta_{n,l_u}^{q-c}Z_{i_1^{(u)},l_u}\cdots Z_{i_c^{(u)},l_u}\right]\right)\right|\\
\leq&\prod_{u = 1}^8 \left\{\E\left(\left[\sum_{l_u = 1}^p\delta_{n,l_u}^{q-c}Z_{i_1^{(u)},l_u}\cdots Z_{i_c^{(u)},l_u}\right]^8\right)\right\}^{1/8} = \E\left(\left[\sum_{l_1 = 1}^p\delta_{n,l_1}^{q-c}Z_{i_1^{(1)},l_1}\cdots Z_{i_c^{(1)},l_1}\right]^8\right)\\
=&\sum_{l_1,...,l_8=1}^p\delta_{n,l_1}^{q-c}...\delta_{n,l_8}^{q-c}\E\left[Z_{i_1^{(1)},l_1}...Z_{i_1^{(1)},l_8}\cdots Z_{i_c^{(1)},l_1}...Z_{i_c^{(1)},l_8}\right] = \sum_{l_1,...,l_8=1}^p\delta_{n,l_1}^{q-c}...\delta_{n,l_8}^{q-c}\left(\E\left[Z_{i_1^{(1)},l_1}...Z_{i_1^{(1)},l_8}\right]\right)^c,
\end{align*}
since $i_1^{(1)},i_2^{(1)},...,i_{c}^{(1)}$ are all different, and $\{Z_i\}$ are i.i.d. Again by H\"older's Inequality,
\begin{align*}
&\sum_{l_1,...,l_8=1}^p\delta_{n,l_1}^{q-c}...\delta_{n,l_8}^{q-c}\left(\E\left[Z_{i_1^{(1)},l_1}...Z_{i_1^{(1)},l_8}\right]\right)^c \\
\leq& \left\{\sum_{l_1,...,l_8=1}^p(\delta_{n,l_1}^{q-c}...\delta_{n,l_8}^{q-c})^{q/(q-c)}\right\}^{(q-c)/q}\left\{\sum_{l_1,...,l_8=1}^p\left(\E\left[Z_{i_1^{(1)},l_1}...Z_{i_1^{(1)},l_8}\right]\right)^{cq/c}\right\}^{c/q}\\
\lesssim& \|\Delta_n\|_q^{8(q-c)}\left\{\sum_{l_1,...,l_8=1}^p\sum_{ \pi}\prod_{B\in\pi}cum(Z_{0,l_i},i \in B)^q \right\}^{c/q}.
\end{align*}
The last line in the above inequalities is due to the CR inequality and the definition of joint cumulants, where $\pi$ runs through the list of all partitions of $\{1, ...,8\}$, $B$ runs through the list of all blocks of the partition $\pi$. As all blocks in a partition are disjoint, we can further bound it as
\begin{align*}
&\|\Delta_n\|_q^{8(q-c)}\left\{\sum_{l_1,...,l_8=1}^p\sum_{ \pi}\prod_{B\in\pi}cum(Z_{0,l_i},i\in B)^q \right\}^{c/q}\\
=&\|\Delta_n\|_q^{8(q-c)}\left\{\sum_{ \pi}\prod_{B\in\pi}\sum_{l_i = 1, i\in B}^pcum(Z_{0,l_i},i\in B)^q \right\}^{c/q} \lesssim \|\Delta_n\|_q^{8(q-c)}\left\{\sum_{ \pi}\|\Sigma\|_q^{q\sum_{B\in \pi}|B|/2}\right\}^{c/q}\\
 \lesssim& \|\Delta_n\|_q^{8(q-c)}\|\Sigma\|_q^{4c},
\end{align*}
where the first inequality in the above is due to Assumption 6.1, A.2, which is a consequence of Assumption 2.1, and the fact that there are only finite number of distinct partitions over $\{1,...,8\}$. This completes the proof of the first result.

For the second result, we first define $A^{\circ n}$ as the notation for the element-wise $n$-th power of any real matrix $A$, i.e. $A^{\circ n}_{i,j} = A_{i,j}^n$. Then we have
\begin{align*}
\left|\sum_{l_1,l_2 = 1}^p\delta_{n,l_1}^{q-c}\delta_{n,l_2}^{q-c}\Sigma_{l_1,l_2}^{c}\right| = \Delta_n^{{\circ (q-c)}^T}\Sigma^{\circ c}\Delta_n^{\circ (q-c)}\leq \|\Delta_n^{\circ (q-c)}\|_2^2\sigma_{\max}(\Sigma^{\circ c}),
\end{align*}
where $\sigma_{\max}$ is the largest eigenvalue. First observe that
$\|\Delta_n^{\circ (q-c)}\|_2^2 = \sum_{l = 1}^p\delta_{n,l}^{2(q-c)} = \|\Delta_n\|_{2(q-c)}^{2(q-c)}$. By properties of $L_q$ norm, $\|\Delta_n\|_{2(q-c)} \leq \|\Delta_n\|_{q}$, if $q \leq 2(q-c)$, and $\|\Delta_n\|_{2(q-c)} \leq p^{1/2(q-c)-1/q}\|\Delta_n\|_{q}$, if $q > 2(q-c)$. This implies $\|\Delta_n\|_{2(q-c)}^{2(q-c)} \leq \max(p^{(2c-q)/q}\|\Delta_n\|_{q}^{2(q-c)}, \|\Delta_n\|_{q}^{2(q-c)})$.

Next, for any symmetric matrix $A$, $\sigma_{\max}(A)\| \leq \|A\|_{\infty} = \max_{i = 1,...,p}\sum_{j = 1}^p|A_{i,j}|$. This, together with Assumption 2.1 (A.2), implies
\begin{align*}
\sigma_{\max}(\Sigma^{\circ c}) \leq \max_{i = 1,...,p}\sum_{j = 1}^p|\Sigma^{\circ c}_{i,j}| \lesssim \max_{i = 1,...,p}\sum_{j = 1}^p(1 \wedge |i-j|)^{-cr} \leq 1 + \sum_{m = 1}^p m^{-cr} \leq \infty,
\end{align*}
for some $r > 2$. This is equivalent to  $\sigma_{\max}(\Sigma^{\circ c}) = O(1)$. Note that $\|\Sigma\|_q^q \geq tr(\Sigma^{\circ q}) \gtrsim p$, which leads to $p^{c/q} \lesssim \|\Sigma\|_q^c$. So
$$\left|\sum_{l_1,l_2 = 1}^p\delta_{n,l_1}^{q-c}\delta_{n,l_2}^{q-c}\Sigma_{l_1,l_2}^{c}\right| \lesssim \max(p^{(2c-q)/q}\|\Delta_n\|_{q}^{2(q-c)}, \|\Delta_n\|_{q}^{2(q-c)}) = o(\|\Delta_n\|_{q}^{2(q-c)}\|\Sigma\|_q^c),$$
since $(2c-q)/q = c/q + (c-q)/q < c/q$, for $c = 1,2,...,q-1$. This completes the proof for the second result. 
	
\end{proof}

This lemma is a generalization to its counterpart in \cite{hdcp}, in which we only have $q=2$. To prove Theorem \ref{Q}, we need the following lemmas to show tightness and finite dimensional convergence.

\begin{lemma}
Under Assumption \ref{cumr}, for any $c=0,1,2...,q$, and define the 3-dimensional index set\\ $\mathcal{G}_n:=\{(i/n,j/n,k/n):i,j,k=0,1,...,n\},$
$$
\E[a_n^{-8}(S_{n,q,c}(r_1;[a_1,b_1])-S_{n,q,c}(r_2;[a_2,b_2]))^8]\le C\|(a_1,r_1,b_1)-(a_2,r_2,b_2)\|^4,
$$
for some constant $C$, any $(a_1,r_1,b_1),(a_2,r_2,b_2)\in\mathcal{G}_n$ such that $\|(a_1,r_1,b_1)-(a_2,r_2,b_2)\|>\delta/n^4$.
\label{mom8}
\end{lemma}
\begin{proof}[Proof of Lemma \ref{mom8}]
By CR-inequality,
\begin{align*}
\E[(S_{n,q,c}(r_1;[a_1,b_1])-S_{n,q,c}(r_2;[a_2,b_2]))^8]\le& C\Big\{ \E[(S_{n,q,c}(r_1;[a_1,b_1])-S_{n,q,c}(r_1;[a_1,b_2]))^8]\\
&+\E[(S_{n,q,c}(r_1;[a_1,b_2])-S_{n,q,c}(r_1;[a_2,b_2]))^8]\\
&+\E[(S_{n,q,c}(r_1;[a_2,b_2])-S_{n,q,c}(r_2;[a_2,b_2]))^8]\Big\}.
\end{align*}
We shall only analyze $\E[(S_{n,q,c}(r;[a,b])-S_{n,q,c}(r;[a,b']))^8]$, and the analysis of the other 2 terms are similar.

Note that for any $a,r,b,b'\in[0,1]$ and $b<b'$, 
\begin{align*}
&\E[(S_{n,q,c}(r;[a,b])-S_{n,q,c}(r;[a,b']))^8]\\
=&\E\left[\left((q-c)\sum_{l=1}^p\sum_{\lf nb\rf+1\le j\le \lf nb'\rf}\sum_{\lf na\rf+1\leq i_1\not=\cdots\not=i_c\leq\lf nr\rf}\sum_{\lf nr\rf+1\leq j_1\not=\cdots\not=j_{q-c-1}\le j-1}\left(\prod_{t=1}^cZ_{i_t,l}\cdot\prod_{s=1}^{q-c-1}Z_{j_s,l}\cdot Z_{j,l}\right)\right)^8\right]\\
=&C\sum_{j^{(\cdot)},i_t^{(\cdot)},j_s^{(\cdot)}}\sum_{l_1,\ldots,l_8=1}^p\prod_{h=1}^8\left(\E\left[\prod_{t=1}^cZ_{i_t^{(h)},l_h}\right]\E\left[\prod_{s=1}^{q-c-1}Z_{j_s^{(h)},l_h}\right]\E\left[Z_{j^{(h)},l_h}\right]\right)\\
\lesssim&n^{4(q-1)}(\lf nb'\rf-\lf nb\rf)^{4}\|\Sigma\|_q^{qh/2}\lesssim n^{4q}\left[(b'-b)^4+\frac{1}{n^4}\right]\|\Sigma\|_q^{4q},
\end{align*}
where we have applied Lemma \ref{sumcum}-(2) to $i_{1}^{(h)},\ldots,i_{c}^{(h)},j_{1}^{(h)},\ldots,j_{q-c-1}^{(h)},j^{(h)}$, and the summation $\sum_{j^{(\cdot)},i_t^{(\cdot)},j_s^{(\cdot)}}$ is over $\lf nb\rf+1\le j^{(h)}\le\lf nb'\rf,\lf na\rf+1\leq i_1^{(h)}\not=\cdots\not=i_c^{(h)}\leq\lf nr\rf,\lf nr\rf+1\leq j_1^{(h)}\not=\cdots\not=j_{q-c-1}^{(h)}\le j^{(h)}-1$ for $h=1,\ldots,8$.
Therefore, we have
$$
a_n^{-8}\E[(S_{n,q,c}(r;[a,b])-S_{n,q,c}(r;[a,b']))^8]\lesssim(b'-b)^4+\frac{1}{n^4}.
$$
\end{proof}

\begin{lemma}
Fix $q, c,$ for any $0 \leq a_{1}<r_{1}<b_{1} \leq 1,0 \leq a_{2}<r_{2}<b_{2},$ any $\alpha_{1}, \alpha_{2} \in \R$, we have
$$
\frac{\alpha_{1}}{a_{n}} S_{n, q, c}\left(r_{1} ;\left[a_{1}, b_{1}\right]\right)+\frac{\alpha_{2}}{a_{n}} S_{n, q, c}\left(r_{2},\left[a_{2}, b_{2}\right]\right) \stackrel{\mathcal{D}}{\longrightarrow} \alpha_{1} Q_{q, c}\left(r_{1} ;\left[a_{1}, b_{1}\right]\right)+\alpha_{2} Q_{q, c}\left(r_{2} ;\left[a_{2}, b_{2}\right]\right),
$$
where
$$ 
\operatorname{cov}\left(Q_{q, c}(r_1;[a_1, b_1]),Q_{q, c}(r_2;[a_2, b_2])\right)=c!(q-c)!(r-A)^c(b-R)^{q-c},
$$
\label{ficonv1}
\end{lemma}
\begin{proof}[Proof of Lemma \ref{ficonv1}]
WLOG, we can assume $a_1<a_2<r_1<r_2<b_1<b_2$. The other terms are similar. Define 
\begin{align*}
\xi_{1,i}=&\frac{q-c}{a_n}\sum_{l=1}^p\sum^*_{\lf na_1\rf+1\leq i_1,\cdots,i_c\leq\lf nr_1\rf}\sum^*_{\lf nr_1\rf+1\leq j_1,\cdots,j_{q-c-1}\le i-1}\left(\prod_{t=1}^cZ_{i_t,l}\cdot\prod_{s=1}^{q-c-1}Z_{j_s,l}\cdot Z_{i,l}\right)\\
\xi_{2,i}=&\frac{q-c}{a_n}\sum_{l=1}^p\sum^*_{\lf na_2\rf+1\leq i_1,\cdots,i_c\leq\lf nr_2\rf}\sum^*_{\lf nr_2\rf+1\leq j_1,\cdots,j_{q-c-1}\le i-1}\left(\prod_{t=1}^cZ_{i_t,l}\cdot\prod_{s=1}^{q-c-1}Z_{j_s,l}\cdot Z_{i,l}\right),\\
\end{align*}
and 
$$
\widetilde{\xi}_{n, i}=\left\{\begin{array}{cc}{\alpha_{1} \xi_{1, i}} & {\text { if }\left\lfloor n r_{1}\right\rfloor+q-c \leq i \leq\left\lfloor n r_{2}\right\rfloor}+q-c-1 \\ {\alpha_{1} \xi_{1, i}+\alpha_{2} \xi_{2, i}} & {\text { if }\left\lfloor n r_{2}\right\rfloor+q-c \leq i \leq\left\lfloor n b_{1}\right\rfloor} \\ {\alpha_{2} \xi_{2, i}} & {\text { if }\left\lfloor n b_{1}\right\rfloor+1 \leq i \leq\left\lfloor n b_{2}\right\rfloor}\end{array}\right..
$$
Define $\mathcal{F}_{i}=\sigma\left(Z_{i}, Z_{i-1}, \cdots\right)$, we can see that under the null $\E[Z_1]=0$, $\widetilde{\xi}_{n,i}$ is a martingale difference sequence w.r.t. $\mathcal{F}_i$, and 
$$
\frac{\alpha_{1}}{a_{n}} S_{n, q, c}\left(r_{1} ;\left[a_{1}, b_{1}\right]\right)+\frac{\alpha_{2}}{a_{n}} S_{n, q, c}\left(r_{2},\left[a_{2}, b_{2}\right]\right)=\sum_{i=\lf nr_1\rf+q-c}^{\lf nb_2\rf}\widetilde{\xi}_{n,i}.    
$$
To apply the martingale CLT (Theorem 35.12 in Billingsley (2008)), we need to verify the following two conditions\\
$\displaystyle(1)\quad\forall\epsilon>0,\sum_{i=\lf nr_1\rf+q-c}^{\lf nb_2\rf}\mathbb{E}\left[\widetilde{\xi}_{n,i}^{2} \textbf{1}\left\{\left|\widetilde{\xi}_{n,i}\right|>\epsilon\right\} \Big| \mathcal{F}_{i-1}\right] \stackrel{p}{\rightarrow} 0$.\\
$\displaystyle(2)\quad V_n=\sum_{i=\lf nr_1\rf+q-c}^{\lf nb_2\rf}\mathbb{E}\left[\widetilde{\xi}_{n,i}^{2}|\F_{i-1}\right] \stackrel{p}{\rightarrow} \sigma^2$.
To prove (1), it suffices to show that
$$
\sum_{i=\lf nr_1\rf+q-c}^{\lf nb_2\rf}\mathbb{E}\left[\widetilde{\xi}_{n,i}^{4}\right]\rightarrow 0.
$$
Observe that
\begin{align*}
&\sum_{i=\lf nr_1\rf+q-c}^{\lf nb_2\rf}\mathbb{E}\left[\widetilde{\xi}_{n,i}^{4}\right]\\
=&\alpha_{1}^{4} \sum_{i=\lf nr_1\rf+q-c}^{\lf nr_2\rf+q-c-1}\mathbb{E}\left[\xi_{1, i}^{4}\right]+\sum_{i=\lf nr_2\rf+q-c}^{\lf nb_1\rf}\mathbb{E}\left[\left(\alpha_{1} \xi_{1, i}+\alpha_{2} \xi_{2, i}\right)^{4}\right]+\alpha_{2}^{4} \sum_{i=\lf nb_1\rf+1}^{\lf nb_2\rf}\mathbb{E}\left[\xi_{2, i}^{4}\right]\\
\le& 8\alpha_1^4\sum_{i=\lf nr_1\rf+q-c}^{\lf nb_1\rf}\mathbb{E}\left[{\xi}_{1,i}^{4}\right]+8\alpha_2^4\sum_{i=\lf nr_2\rf+q-c}^{\lf nb_2\rf}\mathbb{E}\left[{\xi}_{2,i}^{4}\right].
\end{align*}
Straightforward calculations show that
\begin{align*}
&\mathbb{E}\left[{\xi}_{1,i}^{4}\right]\\
&=\frac{C}{n^{2q}\|\Sigma\|_q^{2q}}\sum_{i_t^{(h)},j_s^{(h)}}\sum_{l_1,l_2,l_3,l_4=1}^p\prod_{h=1}^4\left(\E\left[\prod_{t=1}^cZ_{i_t^{(h)},l_h}\right]\E\left[\prod_{s=1}^{q-c-1}Z_{j_s^{(h)},l_h}\right]\E\left[Z_{i,l_h}\right]\right)\\
&\lesssim\frac{1}{n^{2q}\|\Sigma\|_q^{2q}}n^{2(q-1)}\|\Sigma\|_q^{2q}=O(\frac{1}{n^2}).
\end{align*}
The same result holds for $\xi_{2,i}$.
Therefore,
$$
\sum_{i=\lf nr_1\rf+q-c}^{\lf nb_2\rf}\mathbb{E}\left[\widetilde{\xi}_{n,i}^{4}\right]\lesssim\sum_{i=\lf nr_1\rf+q-c}^{\lf nb_1\rf}\mathbb{E}\left[{\xi}_{1,i}^{4}\right]+\sum_{i=\lf nr_2\rf+q-c}^{\lf nb_2\rf}\mathbb{E}\left[{\xi}_{2,i}^{4}\right]=O(\frac{1}{n})\rightarrow 0.
$$
As regards (2), we decompose $V_n$ as follows,
\begin{align*}
&\sum_{i=\lf nr_1\rf+q-c}^{\lf nb_2\rf}\mathbb{E}\left[\widetilde{\xi}_{n,i}^{2}|\F_{i-1}\right]\\
=&\alpha_{1}^{2} \sum_{i=\lf nr_1\rf+q-c}^{\lf nr_2\rf+q-c-1}\mathbb{E}\left[\xi_{1, i}^{2}|\F_{i-1}\right]+\sum_{i=\lf nr_2\rf+q-c}^{\lf nb_1\rf}\mathbb{E}\left[\left(\alpha_{1} \xi_{1, i}+\alpha_{2} \xi_{2, i}\right)^{2}|\F_{i-1}\right]
+\alpha_{2}^{2} \sum_{i=\lf nb_1\rf+1}^{\lf nb_2\rf}\mathbb{E}\left[\xi_{2, i}^{2}|\F_{i-1}\right]\\
=&\alpha_1^2\sum_{i=\lf nr_1\rf+q-c}^{\lf nb_1\rf}\mathbb{E}\left[{\xi}_{1,i}^{2}|\F_{i-1}\right]+\alpha_{2}^2\sum_{i=\lf nr_2\rf+q-c}^{\lf nb_2\rf}\mathbb{E}\left[{\xi}_{2,i}^{2}|\F_{i-1}\right]+2\alpha_{1}\alpha_{2}\sum_{i=\lf nr_2\rf+q-c}^{\lf nb_1\rf}\mathbb{E}\left[\xi_{1, i}\xi_{2, i}|\F_{i-1}\right]\\
=&:\alpha_1^2 V_{1,n}+\alpha_2^2V_{2,n}+2\alpha_1\alpha_2 V_{3,n}.
\end{align*}
We still focus on the case $a_1<a_2<r_1<r_2<b_1<b_2$. Note that 
\begin{align*}
&\sum_{i=\lf nr_1\rf+q-c}^{\lf nb_1\rf}\mathbb{E}\left[{\xi}_{1,i}^{2}|\F_{i-1}\right]\\
=&\frac{(q-c)^2}{n^{q}\|\Sigma\|_q^{q}}c!(q-c-1)!\sum_{i=\lf nr_1\rf+q-c}^{\lf nb_1\rf}\sum_{i_t^{(h)},j_s^{(h)}}\sum_{l_1,l_2=1}^p\Sigma_{l_1l_2}\prod_{h=1}^2\left(\prod_{t=1}^cZ_{i_t^{(h)},l_h}\cdot\prod_{s=1}^{q-c-1}Z_{j_s^{(h)},l_h}\right)\\
=&\frac{(q-c)^2}{n^{q}\|\Sigma\|_q^{q}}c!(q-c-1)!\sum_{i=\lf nr_1\rf+q-c}^{\lf nb_1\rf}\sum_{i_t^{(h)},j_s^{(h)}}^{(1)}\sum_{l_1,l_2=1}^p\Sigma_{l_1l_2}\prod_{h=1}^2\left(\prod_{t=1}^cZ_{i_t^{(h)},l_h}\cdot\prod_{s=1}^{q-c-1}Z_{j_s^{(h)},l_h}\right)\\
&+\frac{(q-c)^2}{n^{q}\|\Sigma\|_q^{q}}c!(q-c-1)!\sum_{i=\lf nr_1\rf+q-c}^{\lf nb_1\rf}\sum_{i_t^{(h)},j_s^{(h)}}^{(2)}\sum_{l_1,l_2=1}^p\Sigma_{l_1l_2}\prod_{h=1}^2\left(\prod_{t=1}^cZ_{i_t^{(h)},l_h}\cdot\prod_{s=1}^{q-c-1}Z_{j_s^{(h)},l_h}\right)\\
=&:V_{1,n}^{(1)}+V_{1,n}^{(2)},
\end{align*}
where $\displaystyle\sum_{i_t^{(h)},j_s^{(h)}}^{(1)}$ denotes the summation over terms s.t. $i_t^{(1)}=i_t^{(2)},j_s^{(1)}=j_s^{(2)},\forall t,s$, and $\displaystyle\sum_{i_t^{(h)},j_s^{(h)}}^{(2)}$ is over the other terms.

It is straightforward to see that $\E[V_{1,n}^{(2)}]=0$ as $Z_i$'s are independent, and

\begin{align*}
\E[V_{1,n}^{(1)}]&=\frac{(q-c)^2}{n^{q}\|\Sigma\|_q^{q}}c!(q-c-1)!n^c(r_1-a_1)^c\sum_{k=1}^{\lf nb_1\rf-\lf nr_1\rf}k^{q-c-1}\sum_{l_1,l_2=1}^p\Sigma_{l_1l_2}^p+o(1)\\
&= c!(q-c)!(r_1-a_1)^c(b_1-r_1)^{q-c}+o(1).
\end{align*}
Note that
\begin{align*}
\E[(V_{1,n}^{(1)})^2]=&\frac{(q-c)^4}{n^{2q}\|\Sigma\|_q^{2q}}[c!(q-c-1)!]^2\sum_{l_1,l_2,l_3,l_4=1}^p\sum_{i=\lf nr_1\rf+q-c}^{\lf nb_1\rf}\sum_{j=\lf nr_1\rf+q-c}^{\lf nb_1\rf}\sum_{i_t^{(h)},j_s^{(h)}}^{*}\Bigg[\Sigma_{l_1l_2}\Sigma_{l_3l_4}\\
&\prod_{h=1}^4\left(\prod_{t=1}^cZ_{i_t^{(h)},l_h}\cdot\prod_{s=1}^{q-c-1}Z_{j_s^{(h)},l_h}\right)\Bigg]+o(1),
\end{align*}
where the summation $\sum_{i_t^{(h)},j_s^{(h)}}^{*}$ is over the range of $i_t^{(h)},j_s^{(h)},h=1,2,3,4$ s.t. $i_t^{(1)}=i_t^{(2)},j_s^{(1)}=j_s^{(2)},i_t^{(3)}=i_t^{(4)},j_s^{(3)}=j_s^{(4)},\forall t,s.$ Note that RHS can be further decomposed into 2 parts. The first part corresponds to the summation of the terms s.t. $\{i^{(h)}_t,j^{(s)}\}$ for $h=1$ and has no intersection with that for $h=3$, which has order
\begin{align*}
&\frac{(q-c)^4}{n^{2q}\|\Sigma\|_q^{2q}}[c!(q-c-1)!]^2n^{2c}(r_1-a_1)^{2c}\sum_{i=1}^{\lf nb_1\rf-\lf nr_1\rf}i^{q-c-1}\sum_{j=1}^{\lf nb_1\rf-\lf nr_1\rf}j^{q-c-1}\sum_{l_1,l_2,l_3,l_4}^p\Sigma_{l_1l_2}^q\Sigma_{l_3l_4}^q\\
=&[c!(q-c)!(r_1-a_1)^c(b_1-r_1)^{q-c}]^2+o(1)=\E^2[V_{1,n}^{(1)}]+o(1).    
\end{align*}
For the second part, it corresponds to the summation of the terms s.t. $\{i^{(h)}_t,j^{(s)}\}$ for $h=1$ and has at least one intersection with that for $h=3$. Since at least one "degree of freedom" for $n$ is lost, the summation still has the form $\sum_{l_1,l_2,l_3,l_4=1}^{p} \mathbb{E}\left[Z_{i_{1}^{(1)}, l_{1}}\cdots Z_{i_{q}^{(1)}, l_{1}}\cdots Z_{i_{1}^{(h)}, l_{h}}\cdots Z_{i_{q}^{(h)}, l_{h}}\right]$ as in Lemma \ref{sumcum}-(2), which has order $O(\|\Sigma\|_q^{2q})$. We can conclude that the second part has order $O(\frac{1}{n})$, and hence goes to 0.

Therefore, $\limsup\left(\E[(V_{1,n}^{(1)})^2]-\E^2[V_{1,n}^{(1)}]\right)\le 0$, which implies $\lim\var(V_{1,n}^{(1)})=0$. Therefore, we can conclude that $V_{1,n}^{(1)}\stackrel{p}{\rightarrow}\lim\E[V_{1,n}^{(1)}]=c!(q-c)!(r_1-a_1)^c(b_1-r_1)^{(q-c)}$. It remains to show that $V_{1,n}^{(2)}\stackrel{p}{\rightarrow}0$. 

It suffices to show that $\E\left[(V_{1,n}^{(2)})^2\right]\rightarrow0$. Based on the same argument as before, by applying Lemma \ref{sumcum}-(2) we know that every kind of summation has the same order $O(\frac{1}{n})$ no matter how $i_t^{(h)},j_s^{(h)},i,j$ intersects with each other. Therefore, the terms in the expansion of $\E\left[(V_{1,n}^{(2)})^2\right]$ for which $n$ has highest degree of freedom should dominate. For these terms, each index in $i_t^{(h)},j_s^{(h)},i,j$ should have exactly one pair. The number of these terms is of order $O(n^{2q})$. The summation has forms $\sum_{l_1,l_2,l_3,l_4=1}^p(\Sigma_{l_1l_2}^d\Sigma_{l_3l_4}^d\Sigma_{l_1l_4}^d\Sigma_{l_2l_3}^e\Sigma_{l_1l_3}^f\Sigma_{l_2l_4}^f)$, s.t. $d>0,e+f>0$ and $d+e+f=q$. We need to show that it is of order $o(\|\Sigma\|_q^{2q})$ to complete the proof. By symmetry, we can assume $e>0$, and therefore $d,e\le1$. Note that for $q>2$,
\begin{align*}
&\sum_{l_1,l_2,l_3,l_4=1}^p(\Sigma_{l_1l_2}^d\Sigma_{l_3l_4}^d\Sigma_{l_1l_4}^e\Sigma_{l_2l_3}^e\Sigma_{l_1l_3}^f\Sigma_{l_2l_4}^f)\\
=&\sum_{l_1,l_2,l_3,l_4=1}^p(\Sigma_{l_1l_2}\Sigma_{l_2l_3}\Sigma_{l_3l_4}\Sigma_{l_4l_1})(\Sigma_{l_1l_2}^{d-1}\Sigma_{l_3l_4}^{d-1}\Sigma_{l_1l_4}^{e-1}\Sigma_{l_2l_3}^{e-1}\Sigma_{l_1l_3}^f\Sigma_{l_2l_4}^f)\\
\le&\left[\sum_{l_1,l_2,l_3,l_4=1}^p|\Sigma_{l_1l_2}\Sigma_{l_2l_3}\Sigma_{l_3l_4}\Sigma_{l_4l_1}|^{q/2}\right]^{2/q}
\left[\sum_{l_1,l_2,l_3,l_4=1}^p|\Sigma_{l_1l_2}^{d-1}\Sigma_{l_3l_4}^{d-1}\Sigma_{l_1l_4}^{e-1}\Sigma_{l_2l_3}^{e-1}\Sigma_{l_1l_3}^f\Sigma_{l_2l_4}^f|^{q/(q-2)}\right]^{1-2/q}\\
\lesssim&o(\|\Sigma\|_q^{4})\cdot\|\Sigma\|_q^{2q-4}=o(\|\Sigma\|_q^{2q}),
\end{align*}
where we have used H\"{o}lder's inequality, along with A.1 and the fact that 
\begin{align*}
&\sum_{l_1,l_2,l_3,l_4=1}^p|\Sigma_{l_1l_2}^{d-1}\Sigma_{l_3l_4}^{d-1}\Sigma_{l_1l_4}^{e-1}\Sigma_{l_2l_3}^{e-1}\Sigma_{l_1l_3}^f\Sigma_{l_2l_4}^f|^{q/(q-2)}\\
\lesssim & \sum_{l_1,l_2,l_3,l_4=1}^p(\Sigma_{l_1l_2}^{q}\Sigma_{l_3l_4}^{q}+\Sigma_{l_1l_3}^{q}\Sigma_{l_2l_4}^{q}+\Sigma_{l_1l_4}^{q}\Sigma_{l_2l_3}^{q})=3\|\Sigma\|_q^{2q}.
\end{align*}
When $q=2$, it must be the case that $d=e=1$, the term becomes $\sum_{l_1,l_2,l_3,l_4=1}^p|\Sigma_{l_1l_2}\Sigma_{l_2l_3}\Sigma_{l_3l_4}\Sigma_{l_4l_1}|$, and directly applying A.1 can yield the desired order.

We can then conclude that $\E[V_{1,n}^{(2)}]\rightarrow0$ and hence $V_{1,n}^{(2)}\stackrel{p}{\rightarrow}0$. Combining what we have proved so far, we obtain $V_{1,n}\stackrel{p}{\rightarrow}c!(q-c)!(r_1-a_1)^c(b_1-r_1)^{q-c}$.

Similar argument shows that  $$V_{2,n}\stackrel{p}{\rightarrow}c!(q-c)!(r_2-a_2)^c(b_2-r_2)^{q-c},\quad V_{3,n}\stackrel{p}{\rightarrow}c!(q-c)!(r_1-a_2)^c(b_1-r_2)^{q-c}.$$
Therefore, we conclude that
\begin{align*}
V_{n}\stackrel{p}{\rightarrow}&\alpha_1^2c!(q-c)!(r_1-a_1)^{c}(b_1-r_1)^{q-c}+\alpha_2^2c!(q-c)!(r_2-a_2)^{c}(b_2-r_2)^{q-c}\\
&+2\alpha_1\alpha_2c!(q-c)!(r_1-a_2)^{c}(b_1-r_2)^{q-c},    
\end{align*}
which completes the proof.
\end{proof}
We can generalize the above lemma to the case when $c_i,q_i$ are not identical.  
\begin{lemma}
Fix $q_1,c_1,q_2,c_2$ for any $0 \leq a_{1}<r_{1}<b_{1} \leq 1,0 \leq a_{2}<r_{2}<b_{2},$ any $\alpha_{1}, \alpha_{2} \in \R$, we have
$$
\frac{\alpha_{1}}{a_{n}} S_{n, q_1, c_1}\left(r_{1} ;\left[a_{1}, b_{1}\right]\right)+\frac{\alpha_{2}}{a_{n}} S_{n, q_2, c_2}\left(r_{2},\left[a_{2}, b_{2}\right]\right) \stackrel{\mathcal{D}}{\longrightarrow}\alpha_{1} Q_{q_1, c_1}\left(r_{1} ;\left[a_{1}, b_{1}\right]\right)+\alpha_{2} Q_{q_2, c_2}\left(r_{2} ;\left[a_{2}, b_{2}\right]\right),
$$
where $Q_{q_1,r_1}$ and $Q_{q_2,r_2}$ are independent Gaussian processes if $q_1\not=q_2$, or $(c_1-c_2)(r_1-r_2) < 0$ or $r_1=r_2,c_1\not=c_2$. And when $q_1=q_2=q,(c_1-c_2)(r_1-r_2)>=0$, we have
$$ 
\operatorname{cov}\left(Q_{q, c_1}(r_1;[a_1, b_1]),Q_{q, c_2}(r_2;[a_2, b_2])\right)=\binom{C}{c}c!(q-C)!(r-A)^c(R-r)^{C-c}(b-R)^{q-C},
$$
\label{ficonv}
\end{lemma}
\begin{proof}[Proof of Lemma \ref{ficonv}]
We use the same notations in proving last lemma, as the proof is similar to the previous one and involves applying martingale CLT, where we have decomposed $V_n$ into 2 parts. Since the argument there can be directly applied, the only additional work is about calculating the mean.

To prove the second statement, we take $c_1<c_2,a_1<a_2<r_1<r_2<b_1<b_2$, as the example case, since the proof for other cases are similar. With the same technique we have used, it can be shown that  
\begin{align*}
E[V_n]\rightarrow&\alpha_1^2c_1!(q-c_1)!(r_1-a_1)^{c_1}(b_1-r_1)^{q-c_1}+\alpha_2^2c_2!(q-c_2)!(r_2-a_2)^{c_2}(b_2-r_2)^{q-c_2}\\
&+2\alpha_1\alpha_2\binom{c_2}{c_1}c_1!(q-c_1)!(r_1-a_2)^{c_1}(r_2-r_1)^{c_2-c_1}(b_1-r_2)^{q-c_2}.    
\end{align*}
To derive the convergence in the statement, we can follow the same argument as before to show the variance goes to 0, and therefore, we have the convergence in distribution, with desired covariance structure. 

As for the first statement, it is straightforward to see that the expectation for the crossing term (corresponding to $\alpha_1\alpha_2$) is 0 for each of the cases in the first statement, which implies that the Gaussian processes have to be independent due to asymptotic normality.
\end{proof}

Now we are ready to complete the proof of Theorem \ref{Q}.
\begin{proof}[Proof of Theorem \ref{Q}]
The tightness is guaranteed by Lemma \ref{mom8} and applying Lemma 7.1 in \cite{kley} with $\Phi(x)=x^4,T=T_n,d(u,u')=\|u-u'\|^{3/4},\bar{\eta}=n^{-3/4}/2$. We omit the detailed proof as the argument is similar to the tightness proof in \cite{hdcp}. Lemma \ref{ficonv} has provided finite dimensional convergence of $S_{n,q,c}$, which has asymptotic covariance structure as $Q_{q,c}$ after normalization. Therefore, we have derived desired process convergence. 
\end{proof}

\begin{proof}[Proof Theorem \ref{power}]
Let $(s,k,m)=(\lf an\rf+1,\lf rn\rf,\lf bn\rf)$ and define
$$
D_{n,q}^Z(r;a,b)=\sum_{l=1}^{p} \sum^*_{s \leq i_{1},\ldots,i_{q} \leq k}\sum^*_{k+1 \leq j_{1}, \ldots ,j_{q} \leq m}\left(Z_{i_{1}, l}-Z_{j_{1}, l}\right) \cdots\left(Z_{i_{q}, l}-Z_{j_{q}, l}\right).$$
Recall that Theorem \ref{convT} holds for $D^Z_{n,q}$ since under the null $D_{n,q}^Z=D_{n,q}$. 

Now we are under the alternative, with the location point $k_1=\lf n\tau_1\rf$ and the change of mean equal to $\Delta_n$. Suppose WLOG $s<k_1<k<m$.
\begin{align*}
D_{n,q}(r;a,b)=&\sum_{l=1}^{p} \sum^*_{s \leq i_{1},\ldots, i_{q} \leq k}\sum^*_{k+1 \leq j_{1}, \ldots,j_{q} \leq m}\left(X_{i_{1}, l}-X_{j_{1}, l}\right) \cdots\left(X_{i_{q}, l}-X_{j_{q}, l}\right)\\
=&q!\sum_{l=1}^{p} \sum_{s \leq i_{1} < \ldots< i_{q} \leq k}\sum^*_{k+1 \leq j_{1},\ldots,j_{q} \leq m}\left(X_{i_{1}, l}-X_{j_{1}, l}\right) \cdots\left(X_{i_{q}, l}-X_{j_{q}, l}\right)\\
=&q!\sum_{l=1}^{p} \sum_{s \leq i_{1} < \ldots< i_{q} \leq k_1}\sum^*_{k+1 \leq j_{1},\ldots, j_{q} \leq m}\left(Z_{i_{1}, l}+\delta_{n,l}-Z_{j_{1}, l}\right) \cdots\left(Z_{i_{q}, l}+\delta_{n,l}-Z_{j_{q}, l}\right)\\
&+q!\sum_{l=1}^{p}\sum_{c=1}^{q-1}\Big[\sum_{s \leq i_{1} < \ldots< i_{c} \leq k_1<i_{c+1}<\ldots<i_q\le k}\sum^*_{k+1 \leq j_{1},\ldots, j_{q} \leq m}\\
&\quad(Z_{i_{1}, l}+\delta_{n,l}-Z_{j_{1}, l}) \cdots(Z_{i_{c}, l}+\delta_{n,l}-Z_{j_{c}, l})(Z_{i_{c+1}, l}-Z_{j_{c+1}, l})\cdots(Z_{i_{q}, l}-Z_{j_{q}, l})\Big]\\
&+q!\sum_{l=1}^{p} \sum_{k_1+1 \leq i_{1} < \ldots< i_{q} \leq k}\sum^*_{k+1 \leq j_{1},\ldots, j_{q} \leq m}\left(Z_{i_{1}, l}-Z_{j_{1}, l}\right) \cdots\left(Z_{i_{q}, l}-Z_{j_{q}, l}\right)\\
=&D_{n,q}^Z+P^{k_1-s+1}_qP^{m-k}_q\|\Delta_n\|_q^q+R_{n,q}.\quad\quad(*)
\end{align*}

First suppose $\gamma_{n,q}\rightarrow\gamma\in[0,\infty)$, which is equivalent to $n^{q/2}\left\|\Delta_{n}\right\|_{q}^q\lesssim\|\Sigma\|_{q}^{q / 2}$. It suffices to show that in this case,
$$
\Big\{n^{-q}a_{n,q}^{-1} D_{n,q}(\cdot;[\cdot,\cdot])\Big\} \leadsto\Big\{ G_{q}(\cdot;[\cdot,\cdot])+\gamma J_q(\cdot;[\cdot,\cdot])\Big\}\text{ in }\ell_{\infty}\left([0,1]^{3}\right).
$$

Since $n^{-q}a_{n,q}^{-1} D^Z_{n,q}(r ;[a, b])$ converges to some non-degenerate process,
and
$$n^{-q}a_{n,q}^{-1}P^{k^*-s+1}_qP^{m-k}_q\|\Delta_n\|_q^q= \gamma(r^*-a)^q(b-r)^q+o(1),$$ 
it remains to show that $n^{-q}a_{n,q}^{-1}R_{n,q}\leadsto 0$.

Note that $R_{n,q}$ consists of terms that are each ratio consistent to 
$$Cn^{2(q-c)}\sum_{l=1}^p\delta_{n,l}^{q-c}D_{n,c,l}(r;a,b),$$
for some constant $C$ depending on $q,a,b,r$ and $c=1,...,q-1$, where
$$
D_{n,c,l}(r;a,b)=\sum_{s \leq i_{1} < \ldots< i_{c} \leq k}\sum^*_{k+1 \leq j_{1}, \ldots, j_{c} \leq m}\left(Z_{i_{1}, l}-Z_{j_{1}, l}\right) \cdots\left(Z_{i_{c}, l}-Z_{j_{c}, l}\right),
$$
which can be further decomposed as 
\begin{align*}
D_{n,c,l}(r;a,b)\asymp&\sum_{d=0}^cC_dn^c\sum^*_{s \leq i_{1} \ldots,i_{d} \leq k}\sum^*_{k+1 \leq j_{1}, \ldots,j_{c-d} \leq m}\left(\prod_{t=1}^{d} Z_{i_{t}, l} \prod_{s=1}^{c-d} Z_{j_{s}, l}\right),
\end{align*}
for some constants depending on $d,c,q$. Therefore, it suffices to show
\begin{align*}
&n^{q-c}a_{n,q}^{-1}\sum_{l=1}^p\delta_{n,l}^{q-c}\sum^*_{s \leq i_{1},\ldots,i_{d} \leq k}\sum^*_{k+1 \leq j_{1}, \ldots, j_{c-d} \leq m}\left(\prod_{t=1}^{d} Z_{i_{t}, l} \prod_{s=1}^{c-d} Z_{j_{s}, l}\right)\\
=&n^{q/2-c}\|\Sigma\|_q^{-q/2}\sum_{l=1}^p\delta_{n,l}^{q-c}\sum^*_{s \leq i_{1}, \ldots, i_{d} \leq k}\sum^*_{k+1 \leq j_{1},\ldots, j_{c-d} \leq m}\left(\prod_{t=1}^{d} Z_{i_{t}, l} \prod_{s=1}^{c-d} Z_{j_{s}, l}\right)
\leadsto 0. 
\end{align*}

Similar argument for showing tightness and finite dimensional convergence in proving Theorem \ref{Q} can be applied. More precisely, we can get a similar moment bound as in Lemma \ref{mom8} and follow the argument there to show the tightness, since we have
\begin{align*}
&n^{4q-8c}\|\Sigma\|_q^{-4q}n^{4c}\sum_{l_{1}, \cdots, l_{8}=1}^{p} \mathbb{E}\left[\delta_{n,l_1}^{q-c}\cdots\delta_{n,l_8}^{q-c}Z_{i_{1}^{(1)}, l_{1}}\cdots Z_{i_{c}^{(1)}, l_{1}}\cdots Z_{i_{1}^{(8)}, l_{h}}\cdots Z_{i_{c}^{(8)}, l_{8}}\right]    \\
=&n^{4(q-c)}\|\Sigma\|_q^{-4q}\sum_{l_{1}, \cdots, l_{8}=1}^{p} \mathbb{E}\left[\delta_{n,l_1}^{q-c}\cdots\delta_{n,l_8}^{q-c}Z_{i_{1}^{(1)}, l_{1}}\cdots Z_{i_{c}^{(1)}, l_{1}}\cdots Z_{i_{1}^{(8)}, l_{8}}\cdots Z_{i_{c}^{(8)}, l_{8}}\right]\\
\lesssim&\|\Delta_n\|_q^{-8(q-c)}\|\Sigma\|_q^{-4c}\sum_{l_{1}, \cdots, l_{8}=1}^{p} \mathbb{E}\left[\delta_{n,l_1}^{q-c}\cdots\delta_{n,l_8}^{q-c}Z_{i_{1}^{(1)}, l_{1}}\cdots Z_{i_{c}^{(1)}, l_{1}}\cdots Z_{i_{1}^{(8)}, l_{8}}\cdots Z_{i_{c}^{(8)}, l_{8}}\right]\lesssim 1,
\end{align*}
by Lemma \ref{sumcum}-(1).

Furthermore, following the proof of Lemma \ref{ficonv1}, Lemma \ref{sumcum}-(3) implies finite dimensional convergence to 0, as
\begin{align*}
&n^{q-2c}\|\Sigma\|_q^{-q}n^c\sum_{l_1,l_2=1}^p\delta_{n,l_1}^{q-c}\delta_{n,l_2}^{q-c}\Sigma_{l_1l_2}^c\\
=&n^{q-c}\|\Sigma\|_q^{-q}\sum_{l_1,l_2=1}^p\delta_{n,l_1}^{q-c}\delta_{n,l_2}^{q-c}\Sigma_{l_1l_2}^c\\
\lesssim&\|\Delta_n\|_q^{-2(q-c)}\|\Sigma\|_q^{-c}\sum_{l_1,l_2=1}^p\delta_{n,l_1}^{q-c}\delta_{n,l_2}^{q-c}\Sigma_{l_1l_2}^c\rightarrow0.
\end{align*}
We have the desired process convergence for $\gamma_{n,q}\rightarrow\gamma<\infty$., which along with the continuous mapping theorem further implies the convergence of the statistic.

When $\gamma=+\infty$, note that $\tilde{T}_{n,q}\geq\frac{U_{n,q}(k_1; 1, n)^{2}}{W_{n,q}(k_1 ; 1, n)}$. Since $k_1$ is the location of the change point, the denominator has the same value as the null. On the contrary, it is immediate to see that the numerator diverges to infinity after normalizing (with $n^{-q}a_{n,q}^{-1}$). Therefore, we have $\tilde{T}_{n,q}\rightarrow+\infty$.
\end{proof}

Before we prove the convergence rate for SN-based estimator, we state the following useful propositions.
\begin{proposition}\label{prop:simple}
For any $1 \leq l < k < m \leq n$, $k \geq l+1$ and $m \geq k+2$, we have:
\begin{enumerate}
\item if $k^* < l$ or $k^* \geq m$, $\U(k;l,m) = \U^Z(k;l,m)$;
\item if $l \leq k \leq k^* < m$, 
\begin{align*}
    \U(k;l,m) =& \U^Z(k;l,m) + (k-l+1)(k-l)(m-k^*)(m - k^*-1)\|\Delta_n\|_2^2 \\&- 2(k-l+1)(m-k^*)(m-k)\sum_{i = l}^k\Delta_n^TZ_i 
    + 2(k-l)(k-l+1)(m-k^*)\sum_{i = k+1}^m\Delta_n^TZ_i\\
    &+ 2(k-l)(m-k^*)\sum_{i = l}^k\Delta_n^TZ_i
    - 2(k-l+1)(k-l+1)\sum_{i = k^*+1}^m\Delta_n^TZ_i;
\end{align*}
\item if $l \leq k^* \leq k < m$,
\begin{align*}
    \U(k;l,m) =& \U^Z(k;l,m) + (k^*-l+1)(k^*-l)(m-k)(m - k-1)\|\Delta_n\|_2^2\\
    &- 2(k^*-l+1)(m-k)(m-k-1)\sum_{i = l}^k\Delta_n^TZ_i + 2(m-k-1)(k^*-l+1)(k-l+1)\sum_{i = k+1}^m\Delta_n^TZ_i \\
    &+ 2(m-k-1)(m-k)\sum_{i = l}^{k^*}\Delta_n^TZ_i- 2(m-k-1)(k^*-l+1)\sum_{i = k+1}^m\Delta_n^TZ_i.
\end{align*}
\end{enumerate}
\end{proposition}
Let $\epsilon_n = n\gamma_{n,2}^{-1/4 + \kappa}$. We have the following result.
\begin{proposition}\label{prop:UW}
Under Assumption \ref{ass}, 
\begin{enumerate}
    \item $P\left(\sup_{k \in \Omega_n}\U(k;1,n)^2 - \U(k^*;1,n)^2 \geq 0\right)\rightarrow 0$;\\
    \item $P\left(W_{n,2}(k^*;1,n) - \inf_{k \in \Omega_n}W_{n,2}(k;1,n) \geq 0\right) \rightarrow 0$,
\end{enumerate}
where $\Omega_n = \{k: |k - k^*| > \epsilon_n\}$.
\end{proposition}

Now we are ready to prove the convergence rate for SN-based statistic $\hat \tau$.
\begin{proof}[Proof of Theorem \ref{thm:consistency}]
Due to the fact that $\hat{k}$ is the global maximizer, we have 
\begin{align*}
    0 &\leq \frac{\U(\hat{k};1,n)^2}{W_{n,2}(\hat{k};1,n)} - \frac{\U(k^*;1,n)^2}{W_{n,2}(k^*;1,n)}\\
    &= \frac{\U(\hat{k};1,n)^2}{W_{n,2}(\hat{k};1,n)} - \frac{\U(k^*;1,n)^2}{W_{n,2}(\hat{k};1,n)} + \frac{\U(k^*;1,n)^2}{W_{n,2}(\hat{k};1,n)} - \frac{\U(k^*;1,n)^2}{W_{n,2}(k^*;1,n)}\\
    &= \frac{1}{W_{n,2}(\hat{k};1,n)}(\U(\hat{k};1,n)^2 - \U(k^*;1,n)^2) + \frac{\U(k^*;1,n)^2}{W_{n,2}(\hat{k};1,n)W_{n,2}(k^*;1,n)}(W_{n,2}(k^*;1,n) - W_{n,2}(\hat{k};1,n)).\\
\end{align*}

Since $\U(k^*;1,n)^2$, $W_{n,2}(\hat{k};1,n)$ and $W_{n,2}(k^*;1,n)$ are all strictly positive almost surely, we can then conclude that $\U(\hat{k};1,n)^2 - \U(k^*;1,n)^2 \geq 0$ or $W_{n,2}(k^*;1,n) - W_{n,2}(\hat{k};1,n) \geq 0$. Define $\Omega_n = \{k: |k - k^*| > \epsilon_n\}$. If $\hat{k} \in \Omega_n$, then there exists at least one $k \in \Omega_n$ such that  $\U(k;1,n)^2 - \U(k^*;1,n)^2 \geq 0$ or $W_{n,2}(k^*;1,n) - W_{n,2}(k;1,n) \geq 0$. This implies
$$P(\hat{k} \in \Omega_n) \leq P\left(\sup_{k \in \Omega_n}\U(k;1,n)^2 - \U(k^*;1,n)^2 \geq 0\right) + P\left(W_{n,2}(k^*;1,n) - \inf_{k \in \Omega_n}W_{n,2}(k;1,n) \geq 0\right).$$

By Proposition \ref{prop:UW}, it is straightforward to see that $P(\hat{k} \in \Omega_n) \rightarrow 0$, and this completes the proof.
\end{proof}


\begin{proof}[Proof of Proposition \ref{prop:simple}]
If $k^* < l$ or $k^* \geq m$, then $\E[X_i]$ are all identical, for $i = l,...,m$. This implies that $\U(k;l,m) = \sum_{l \leq i_1 \neq i_2 \leq k}\sum_{k+1 \leq j_1 \neq j_2 \leq m}(X_{i_1} - X_{j_1})^T(X_{i_1} - X_{j_2}) = \sum_{l \leq i_1 \neq i_2 \leq k}\sum_{k+1 \leq j_1 \neq j_2 \leq m}(Z_{i_1} - Z_{j_1})^T(Z_{i_1} - Z_{j_2}) = \U^Z(k;l,m)$.

When $l \leq k^* < m$, there are two scenarios depending on the value of $k$. If $k \leq k^*$, note that $\E[X_i] = \Delta_n$ for any $i > k^*$ and zero otherwise, then by straightforward calculation we have
\begin{align*}
    &\U(k;l,m) = \sum_{l \leq i_1 \neq i_2 \leq k}\sum_{k+1 \leq j_1 \neq j_2 \leq m}(X_{i_1} - X_{j_1})^T(X_{i_1} - X_{j_2})\\
    =&\sum_{l \leq i_1 \neq i_2 \leq k}\sum_{k+1 \leq j_1 \neq j_2 \leq m}(Z_{i_1} - Z_{j_1} - \E[X_{j_1}])^T(Z_{i_1} - Z_{j_2} - \E[X_{j_2}])\\
    =&\U(k;l,m) + (k-l+1)(k-l)(m-k^*)(m-k^*-1)\|\Delta_n\|_2^2 - 2(k-l)(m-k^*)\sum_{i = l}^k\sum_{j = k+1}^{k^*}\Delta_n^T(Z_{i} - Z_j)\\
    &-2(k-l)(m-k^*-1)\sum_{i = l}^k\sum_{j = k^*+1}^m\Delta_n^T(Z_i - Z_j)\\
    =&\U^Z(k;l,m) + (k-l+1)(k-l)(m-k^*)(m - k^*-1)\|\Delta_n\|_2^2 - 2(k-l)(m-k^*)(m-k)\sum_{i = l}^k\Delta_n^TZ_i \\
    &+ 2(k-l)(m-k^*)(k-l+1)\sum_{i = k+1}^m\Delta_n^TZ_i + 2(k-l)(m-k^*)\sum_{i = l}^k\Delta_n^TZ_i\\
    &- 2(k-l)(k-l+1)\sum_{i = k^*+1}^m\Delta_n^TZ_i.
\end{align*}

Similarly if $k \geq k^*$ we have
\begin{align*}
    &\U(k;l,m) = \sum_{l \leq i_1 \neq i_2 \leq k}\sum_{k+1 \leq j_1 \neq j_2 \leq m}(X_{i_1} - X_{j_1})^T(X_{i_1} - X_{j_2})\\
    =&\sum_{l \leq i_1 \neq i_2 \leq k}\sum_{k+1 \leq j_1 \neq j_2 \leq m}(Z_{i_1} - Z_{j_1} + \E[X_{i_1}] - \Delta_n)^T(Z_{i_1} - Z_{j_2} + \E[X_{i_2}] - \Delta_n)\\
    =&\U(k;l,m) + (k^*-l+1)(k^*-l)(m-k)(m-k-1)\|\Delta_n\|_2^2 - 2(m-k-1)(k^*-l)\sum_{i = l}^{k^*}\sum_{j = k+1}^{m}\Delta_n^T(Z_{i} - Z_j)\\
    &-2(m-k-1)(k^*-l+1)\sum_{i = k^*+1}^k\sum_{j = k+1}^m\Delta_n^T(Z_i - Z_j)\\
    =& \U^Z(k;l,m) + (k^*-l+1)(k^*-l)(m-k)(m - k-1)\|\Delta_n\|_2^2 - 2(k^*-l+1)(m-k)(m-k-1)\sum_{i = l}^k\Delta_n^TZ_i \\
    &+ 2(m-k-1)(k^*-l+1)(k-l+1)\sum_{i = k+1}^m\Delta_n^TZ_i + 2(m-k-1)(m-k)\sum_{i = l}^{k^*}\Delta_n^TZ_i\\
    &- 2(m-k-1)(k^*-l+1)\sum_{i = k+1}^m\Delta_n^TZ_i.
\end{align*}
\end{proof}

\begin{proof}[Proof of Proposition \ref{prop:UW}]
To show the first result, we first assume $k < k^* - \epsilon_n$. Then according to Proposition \ref{prop:simple},
\begin{align*}
    \U(k;1,n) &= \U^Z(k;1,n) + k(k-1)(n-k^*)(n-k^*-1)\|\Delta_n\|_2^2 - 2(k-1)(n-k^*)(n-k)\sum_{i = 1}^k\Delta_n^TZ_i\\
    &+2k(k-1)(n-k^*)\sum_{i = k+1}^n\Delta_n^TZ_i + 2(k-1)(n-k^*)\sum_{i = 1}^k\Delta_n^TZ_i-2k(k-1)\sum_{i = k^* + 1}^n\Delta_n^TZ_i.
\end{align*}
Similarly we have
\begin{align*}
    \U(k^*;1,n) &= \U^Z(k^*;1,n) + k^*(k^*-1)(n-k^*)(n-k^*-1)\|\Delta_n\|_2^2 - 2(k^*-1)(n-k^*)(n-k^*-1)\sum_{i = 1}^{k^*}\Delta_n^TZ_i\\
                &+2k^*(k^*-1)(n-k^*-1)\sum_{i = k^*+1}^n\Delta_n^TZ_i.\\
\end{align*}
It is easy to verify that $\E[\U(k;1,n)] = k(k-1)(n-k^*)(n-k^*-1)\|\Delta_n\|_2^2$, for $k \leq k^*$. Furthermore, by Theorem 2.1 in \cite{hdcp} and the argument therein, we have 
$$\sup_{k = 2,...,n-2}|\U^Z(k;1,n)| = O(n^3\|\Sigma\|_F) = o_p(n^{3.5}\sqrt{\|\Sigma\|_F}\|\Delta_n\|_2),$$
since $\sqrt{\|\Sigma\|_F} = o(\sqrt{n}\|\Delta_n\|_2)$ by Assumption \ref{ass} (3), and 
$$\sup_{1\leq a \leq b \leq n}\left|\sum_{i = a}^b\Delta_n^TZ_i\right| = O_p(\sqrt{n}\sqrt{\Delta_n^T\Sigma\Delta_n}) \leq O_p(\sqrt{n\|\Sigma\|_2}\|\Delta_n\|_2) \leq O_p(\sqrt{n\|\Sigma\|_F}\|\Delta_n\|_2).$$ 
These imply that 
\begin{align*}
\U(k^*;1,n)=&k^*(k^*-1)(n-k^*)(n-k^*-1)\|\Delta_n\|_2^2 + O_p(n^{3.5}\|\Delta_n\|_2\sqrt{\|\Sigma\|_F})\\
=&k^*(k^*-1)(n-k^*)(n-k^*-1)\|\Delta_n\|_2^2 + 
o_p(n^4\|\Delta_n\|_2^2),
\end{align*} 
since $\sqrt{\|\Sigma\|_F} = o(\sqrt{n}\|\Delta_n\|_2)$ by Assumption \ref{ass} (3). Therefore, we have 
$$P(\sup_{k < k^* - \epsilon_n}|\U(k;1,n)| + \U(k^*;1,n) > 0) \rightarrow 1.$$

In addition, 
\begin{align*}
   &\sup_{k < k^* - \epsilon_n}|\U(k;1,n)| -  \U(k^*;1,n)\\
   \leq& \sup_{k < k^* - \epsilon_n}k(k-1)(n-k^*)(n-k^*-1)\|\Delta_n\|_2^2 + O_p(n^{3.5}\|\Delta_n\|_2\sqrt{\|\Sigma\|_F})\\
   &- k^*(k^*-1)(n-k^*)(n-k^*-1)\|\Delta_n\|_2^2 
   - O_p(n^{3.5}\|\Delta_n\|_2\sqrt{\|\Sigma\|_F})\\
   =&  -\epsilon_n(2k^* - \epsilon_n - 1)(n-k^*)(n-k^*-1)\|\Delta_n\|_2^2 + O_p(n^{3.5}\|\Delta_n\|_2\sqrt{\|\Sigma\|_F})\\
   =&-\epsilon_n(2k^* - \epsilon_n - 1)(n-k^*)(n-k^*-1)\|\Delta_n\|_2^2 + O_p(n^{4}\|\Delta_n\|_2^2/\sqrt{\gamma_{n,2}}).
\end{align*}
Since ${n/\sqrt{\gamma_{n,2}}} = o(n\gamma_{n,2}^{-1/4 + \kappa}) = o(\epsilon_n)$, we have 
\begin{align*}
&P(\sup_{k < k^* - \epsilon_n}|\U(k;1,n)|- \U(k^*;1,n) < 0)\\
\geq & P\Big(-\epsilon_n(2k^* - \epsilon_n - 1)(n-k^*)(n-k^*-1)\|\Delta_n\|_2^2 + O_p(n^{4}\|\Delta_n\|_2^2/\sqrt{\gamma_{n,2}}) < 0\Big) \rightarrow 1.
\end{align*}

Finally, it is straightforward to see that 
\begin{align*}
    &\sup_{k \leq k^* - \epsilon_n}\U(k;1,n)^2 - \U(k^*;1,n)^2 \leq \left(\sup_{k \leq k^* - \epsilon_n}|\U(k;1,n)|\right)^2 - \U(k^*;1,n)^2\\
    =&\left(\sup_{k \leq k^* - \epsilon_n}|\U(k;1,n)| - \U(k^*;1,n)\right)\left(\sup_{k \leq k^* - \epsilon_n}|\U(k;1,n)| + \U(k^*;1,n)\right).
\end{align*}
And
\begin{align*}
    &P\left(\left(\sup_{k \leq k^* - \epsilon_n}|\U(k;1,n)| - \U(k^*;1,n)\right)\left(\sup_{k \leq k^* - \epsilon_n}|\U(k;1,n)| + \U(k^*;1,n)\right) <  0\right)\\
    \geq& P\left(\left\{\sup_{k \leq k^* - \epsilon_n}|\U(k;1,n)| - \U(k^*;1,n) < 0\right\}\bigcap\left\{\sup_{k \leq k^* - \epsilon_n}|\U(k;1,n)| + \U(k^*;1,n) >  0\right\}\right) \rightarrow 1,
\end{align*}
since both $P(\sup_{k < k^* - \epsilon_n}|\U(k;1,n)| + \U(k^*;1,n) > 0)$ and $P(\sup_{k < k^* - \epsilon_n}|\U(k;1,n)|- \U(k^*;1,n) < 0)$ converge to 1. This is equivalent to 
$$P\left(\left(\sup_{k \leq k^* - \epsilon_n}|\U(k;1,n)| - \U(k^*;1,n)\right)\left(\sup_{k \leq k^* - \epsilon_n}|\U(k;1,n)| + \U(k^*;1,n)\right) \geq  0\right) \rightarrow 0,$$
and it implies that $P\left(\sup_{k < k^* - \epsilon_n}\U(k;1,n)^2 - \U(k^*;1,n)^2 \geq 0\right)\rightarrow 0$. Similar tactics can be applied to the case $k > k^* + \epsilon_n$ and by combining the two parts we have $P\left(\sup_{k \in \Omega_n}\U(k;1,n)^2 - \U(k^*;1,n)^2 \geq 0\right)\rightarrow 0$. Therefore this completes the proof for the first result.

	It remains to show the second part. Let us again assume $k < k^* - \epsilon_n$ first. By Proposition \ref{prop:simple} we have
    \begin{align*}
        W_{n,2}(k^*;1,n) &= \frac{1}{n}\sum_{t = 2}^{k^*-2}\U(t;1,k^*)^2 + \frac{1}{n}\sum_{t = k^*+2}^{n-2}\U(t;k^*+1,n)^2\\
        &=\frac{1}{n}\sum_{t = 2}^{k^*-2}\U^Z(t;1,k^*)^2 + \frac{1}{n}\sum_{t = k^*+2}^{n-2}\U^Z(t;k^*+1,n)^2,
    \end{align*}
    and
    \begin{align*}
        W_{n,2}(k;1,n) &= \frac{1}{n}\sum_{t = 2}^{k-2}\U(t;1,k)^2 + \frac{1}{n}\sum_{t = k+2}^{n-2}\U(t;k+1,n)^2\\
        &=\frac{1}{n}\sum_{t = 2}^{k-2}\U^Z(t;1,k)^2 + \frac{1}{n}\sum_{t = k+2}^{n-2}\U(t;k+1,n)^2.
    \end{align*}
    
    When $t$ is between $k+2$ and $k^*$, by Proposition \ref{prop:simple} we have
    \begin{align*}
        \U(t;k+1,n) = &\U^Z(t;k+1,n) + (t-k)(t-k-1)(n-k^*)(n - k^*-1)\|\Delta_n\|_2^2\\
        &- 2(t-k-1)(n-k^*)(n-t)\sum_{i = k+1}^t\Delta_n^TZ_i+ 2(t-k-1)(n-k^*)(t-k)\sum_{i = t+1}^n\Delta_n^TZ_i \\
        &+ 2(t-k-1)(n-k^*)\sum_{i = k+1}^t\Delta_n^TZ_i
        - 2(t-k-1)(t-k)\sum_{i = k^*+1}^n\Delta_n^TZ_i,
    \end{align*}
    and from the above decomposition we observe that $\E[\U(t;k+1,n)] = (t-k)(t-k-1)(n-k^*)(n - k^*-1)\|\Delta_n\|_2^2$, which is the second term in the above equality. Then
    \begin{align*}
        \U(t;k+1,n)^2 &= (\U(t;k+1,n) - \E[\U(t;k+1,n)] + \E[\U(t;k+1,n)])^2\\
        &\geq\E[\U(t;k+1,n)]^2 + 2\E[\U(t;k+1,n)](\U(t;k+1,n) - \E[\U(t;k+1,n)])\\
        &\geq\E[\U(t;k+1,n)]^2 - 2\E[\U(t;k+1,n)]\sup_{t = k+2,...,n-2}|\U(t;k+1,n) - \E[\U(t;k+1,n)]|,\\
        \end{align*}
    since $\E[\U(t;k+1,n)] > 0$. Furthermore, 
    \begin{align*}
        &\sup_{t = k+2,...,n-2}|\U(t;k+1,n) - \E[\U(t;k+1,n)]|\\ \leq & \sup_{t = k+2,...,n-2}|\U^Z(t;k+1,n)| + 8n^3\sup_{a < b, a,b = 1,...,n}\left|\sum_{i = a}^b\Delta_n^TZ_i\right|\\
        =&O_p(n^3\|\Sigma\|_F) + O_p(n^{3.5}\sqrt{\Delta_n^T\Sigma\Delta_n}) = o_p(n^4\|\Delta_n\|^2/\sqrt{a_n}),
    \end{align*}
    due to Assumption \ref{ass}, Theorem 2.1 and the argument in \cite{hdcp}. 
    
Similarly when $t$ is between $k^*$ and $n-2$, we have
\begin{align*}
    \U(t;k+1,n)^2 \geq\E[\U(t;k+1,n)]^2 - 2\E[\U(t;k+1,n)]\sup_{t = k+2,...,n-2}|\U(t;k+1,n) - \E[\U(t;k+1,n)]|,
\end{align*}
where $\E[\U(t;k+1,n)] = (k^*-k)(k^*-k-1)(n-t)(n-t-1)\|\Delta_n\|_2^2 > 0$, and 
\begin{align*}
 &\sup_{t = k+2,...,n-2}|\U(t;k+1,n) - \E[\U(t;k+1,n)]| \\
 \leq& O_p(n^3\|\Sigma\|_F) + O_p(n^{3.5}\sqrt{\Delta_n^T\Sigma\Delta_n}) = O_p(n^4\|\Delta_n\|^2/\sqrt{a_n})  
\end{align*}

Therefore by combining the above results we obtain that
\begin{align*}
    &W_{n,2}(k;1,n) \\
    \geq& \frac{1}{n}\sum_{t = 2}^{k-2}\U^Z(t;1,k)^2 + \frac{1}{n}\sum_{t = k+2}^{k^*}\E[\U(t;k+1,n)]^2 + \frac{1}{n}\sum_{t = k^*+1}^{n-2}\E[\U(t;k+1,n)]^2\\
    &- \frac{2}{n}\sup_{t = k+2,...,n-2}|\U(t;k+1,n) - \E[\U(t;k+1,n)]|\sum_{t = k+2}^{k^*}\E[\U(t;k+1,n)]\\
    &- \frac{2}{n}\sup_{t = k+2,...,n-2}|\U(t;k+1,n) - \E[\U(t;k+1,n)]|\sum_{t = k^*+1}^{n-2}\E[\U(t;k+1,n)]\\
    \gtrsim& (k^*-k)^5n^3\|\Delta_n\|_2^4  - (k^*-k)^3n\|\Delta_n\|_2^2\sup_{t = k+2,...,n-2}|\U(t;k+1,n) - \E[\U(t;k+1,n)]|\\
    &+(k^*-k)^4n^4\|\Delta_n\|_2^4 - (k^*-k)^2n^2\|\Delta_n\|_2^2\sup_{t = k+2,...,n-2}|\U(t;k+1,n) - \E[\U(t;k+1,n)]|\\ 
    &- \left(\sup_{k}\sup_{t = 2,...,k-2}|\U^Z(t;1,k)|\right)^2\\
    =& (k^*-k)^3n^3\|\Delta_n\|_2^4[(k^*-k)^2 - o_p(n^2/\sqrt{\gamma_{n,2}})] + (k^*-k)^2n^4\|\Delta_n\|_2^4[(k^*-k)^2 - o_p(n^2/\sqrt{\gamma_{n,2}})] - O_p(n^6\|\Sigma\|_F^2)\\
    \geq&(k^*-k)^3n^3\|\Delta_n\|_2^4[\epsilon_n^2 - o_p(n^2/\sqrt{\gamma_{n,2}})] + (k^*-k)^2n^4\|\Delta_n\|_2^4[\epsilon_n^2 - o_p(n^2/\sqrt{\gamma_{n,2}})] - O_p(n^6\|\Sigma\|_F^2)\\
    =&((k^*-k)^3n^3 + (k^*-k)^2n^4)\|\Delta_n\|_2^4\epsilon_n^2(1 - o_p(1))- O_p(n^6\|\Sigma\|_F^2),
\end{align*}
since $\epsilon_n = na_n^{-1/4 + \kappa}$. And
\begin{align*}
    \inf_{k < k^* - \epsilon}W_{n,2}(k;1,n) \gtrsim (\epsilon_n^3n^3 + \epsilon_n^2n^4)\|\Delta_n\|_2^4\epsilon_n^2(1 - o_p(1))- O_p(n^6\|\Sigma\|_F^2) = \epsilon_n^4n^4\|\Delta_n\|^4(1 - o_p(1)),
\end{align*}
since $\epsilon_n = o(n)$ and $\epsilon_n^4n^4\|\Delta_n\|^4/(n^6\|\Sigma\|_F^2) = \gamma_{n,2}^{1 + 4\kappa} \rightarrow \infty$. By very similar arguments, we can obtain the same bound for $\inf_{k > k^* + \epsilon}W_{n,2}(k;1,n)$, and hence $\inf_{k \in \Omega_n}W_{n,2}(k;1,n) \gtrsim  \epsilon_n^4n^4\|\Delta_n\|^4(1 - o_p(1))$. On the other hand, Theorem 2.1 implies that $W_{n,2}(k^*;1,n) = \frac{1}{n}\sum_{t = 2}^{k^*-2}\U^Z(t;1,k^*)^2 + \frac{1}{n}\sum_{t = k^*+2}^{n-2}\U^Z(t;k^*+1,n)^2 = O_p(n^6\|\Sigma\|_F^2)$. This indicates that $W_{n,2}(k^*;1,n) = \epsilon_n^4n^4\|\Delta_n\|^4o_p(1)$, and consequently, 
$$P\left(W_{n,2}(k^*;1,n) - \inf_{k \in \Omega_n}W_{n,2}(k;1,n) \geq 0\right) \leq P\Big(\epsilon_n^4n^4\|\Delta_n\|^4o_p(1) - \epsilon_n^4n^4\|\Delta_n\|^4(1 - o_p(1)) \geq 0\Big) \rightarrow 0.$$ This completes the whole proof.
\end{proof}

\section{Application to network change-point detection}
\label{sec:network}

Our change-point testing and estimation methods are applicable to network change-point detection in the following sense. Suppose we observe $n$ independent networks $\{A_t\}_{t=1}^n$ over time with $m$ nodes. Here $A_t$ is the $m\times m$ adjacency matrix at time $t$. We assume the edges in each network are generated from Bernoulli random variables and are un-directed. That is, 
\[A_{ij,t}=1~ \mbox{if nodes $i$ and $j$ are connected at time $t$ and} ~ 0 ~\mbox{otherwise}.\]
Let $A_t=(A_{ij,t})_{i,j=1}^{m}$ and assume $E(A_{ij,t})=p_{ij,t}$. Let $E(A_t)=\Theta_t=(p_{ij,t})_{i,j=1}^{m}$.

Suppose that  we are interested in testing 
\[H_0:\Theta_1=\cdots=\Theta_n\]
versus certain change point alternatives. Here we can convert the adjacency matrix into a high-dimensional vector, and apply our test and estimation procedures. Note that  a mean shift in $vech(\Theta_t)$ implies a shift in variance matrix of $vech(A_t)$, so the variance matrix is not constant under the alternative.
However, the asymptotic distribution of our SN-based test statistics still holds under the null,  and our change-point detection method is applicable. Note that our method allows the edges to be weakly dependent, which can be satisfied by many popular network models; see \cite{wang2018opt}.


To examine the finite sample performance of our change-point testing and estimation in the network framework, we consider the following stochastic block model as in \cite{wang2018opt}. We generate $A_t$ as a matrix with entries being  i.i.d. Bernoulli variables with mean matrix $\Theta_t=\mu_t ZQZ^T-\diag(\mu_t ZQZ^T)$ where $Z\in \R^{m\times r}$ is the membership matrix and  $Q\in [0, 1]^{r\times r}$ is the connectivity matrix. We set $Z$ to be the first $r$ columns of identity matrix $I_m$ so that $\operatorname{rank}(Z)=r$, and $Q=\bm{1}_r\cdot\bm{1}_{r}^T$ be a matrix of ones.

Table \ref{sizeonenw} presents the size with 1000 Monte Carlo repetitions. We take $r=cm,\mu_t\equiv0.1/c$ with $c=0.2,1$.

\begin{table}[H]
\scriptsize
\centering
\begin{tabular}{|c|c|c|c|c|c|c|c|c|c|c|c|}
  \hline
  {DGP} & \multirow{2}{*}{$(n,m)$}  &  \multicolumn{5}{c|}{ $\mathcal{H}_0$,5\%} &  \multicolumn{5}{c|}{ $\mathcal{H}_0$,10\%} \\
  \cline{3-12}
  $c$ & & $q=2$ & $q=4$ & $q=6$ & $q=2,4$ & $q=2,6$ & $q=2$ & $q=4$& $q=6$& $q=2,4$ & $q=2,6$ \\
   \hline
  \multirow{2}{*}{1} & (200,10) & 0.035 & 0.096 & 0.068 & 0.08 & 0.048 & 0.075 & 0.152  & 0.135 & 0.124 & 0.096   \\
  \cline{2-12}
  & (400,20) & 0.054 & 0.084 & 0.049 & 0.071 & 0.048 & 0.097 & 0.142 & 0.094 & 0.135 & 0.099\\
  \hline
    \multirow{2}{*}{0.2} & (200,10) & 0.065 & 0.117 & 0.08 & 0.116 & 0.062 & 0.095 & 0.153  & 0.151 & 0.147 & 0.121   \\
  \cline{2-12}
  & (400,20)  & 0.05 & 0.101 & 0.043 & 0.09 & 0.047 & 0.099 & 0.153 & 0.096 & 0.137 & 0.083\\
  \hline
\end{tabular}
\caption{Size for testing one change point of network time series}
\label{sizeonenw}
\end{table}

As regards the power simulation, we generate the network data with a change point located at $\lf n/2\rf$, which leads to $\mu_t=\mu+\delta\ind(t>n/2)\cdot\mu$. We take $\mu=0.1/c, r=cm$ with $c=0.2,1$ and $\delta=0.2,0.5$. We obtain the empirical power based on 1000 Monte Carlo repetitions.

\begin{table}[H]
\scriptsize
\centering
\begin{tabular}{|c|c|c|c|c|c|c|c|c|c|c|c|}
  \hline
  DGP & \multirow{2}{*}{$(n,m)$}  &  \multicolumn{5}{c|}{ $\mathcal{H}_0$,5\%} &  \multicolumn{5}{c|}{ $\mathcal{H}_0$,10\%} \\
  \cline{3-12}
   ($\delta,c$) & & $q=2$ & $q=4$ & $q=6$ & $q=2,4$ & $q=2,6$  & $q=2$ & $q=4$& $q=6$& $q=2,4$ & $q=2,6$  \\
   \hline
  \multirow{2}{*}{(0.2,1)} & (200,10) & 0.152 & 0.172 & 0.116 & 0.19 & 0.145 & 0.223 & 0.254 & 0.225 & 0.265 & 0.222  \\
  \cline{2-12}
  & (400,20)  & 0.83 & 0.309 & 0.238 & 0.787 & 0.775 & 0.908 & 0.411 & 0.364 & 0.865 & 0.85  \\
  \hline
    \multirow{2}{*}{(0.5,1)} & (200,10) & 0.93 & 0.628 & 0.527 & 0.917 & 0.904 & 0.963 & 0.723 & 0.666 & 0.952 & 0.937  \\
  \cline{2-12}
  & (400,20)  & 1 & 0.995 & 0.97 & 1 & 1 & 1 & 0.997 & 0.99 & 1 & 1  \\
  \hline
    \multirow{2}{*}{(0.2,0.2)} & (200,10) & 0.804 & 0.677 & 0.61 & 0.798 & 0.755  & 0.866 & 0.75 & 0.708 & 0.86 & 0.829  \\
  \cline{2-12}
  & (400,20)  & 1 & 0.994 & 0.991 & 1 & 1 & 1 & 0.997 & 0.999 & 1 & 1 \\
  \hline
    \multirow{2}{*}{(0.5,0.2)} & (200,10) & 1 & 1 & 1 & 1 & 1 & 1 & 1 & 1 & 1 & 1   \\
  \cline{2-12}
  & (400,20)  & 1 & 1 & 1 & 1 & 1 & 1 & 1 & 1 & 1 & 1  \\
  \hline
\end{tabular}
\caption{Power for testing one change point of network time series}
\label{poweronenw}
\end{table}
We can see that our method exhibits similar size behavior as compared to the setting for Gaussian distributed data in Section~\ref{sec:sim1}.  The power also appears to be quite good and increases when the signal increases. Unfortunately, we are not aware of any particular testing method tailored for single network change-point so we did not include any other method into the comparison. 

To estimate the change-points in the network time series, we also combine our method with WBS. We generate 100 samples of networks with connection probability $\mu_t$ and sparsity parameter $r$. The 3 change points are located at $30, 60$ and $90$. We take $\mu_t=\mu+\delta\cdot\ind(30<t\le60\text{ or }t>90)\cdot\mu$. We report the MSE and ARI of 100 Monte Carlo simulations as before. We compare our method with modified neighborhood smoothing (MNBS) algorithm in \cite{zhao2019} and the graph-based test in \cite{chen2015graph} combined with the binary segmentation (denoted as CZ). We do not include a comparison with \cite{wang2018opt} as their method requires two iid samples. We can see that CZ performs worse than the other two methods as our simulation involves non-monotonic changes in the mean that does not favor binary segmentation. When the network becomes sparse, i.e. $c=0.3$, our method also has better performance than MNBS. Overall the performance of our method (e.g., WBS-SN(2), WBS-SN(2,6)) seem quite stable. Of course, the scope of this simulation is quite limited, and we leave a more in-depth investigation of network change-point estimation to near future.

\begin{table}[H]
    \centering
    \begin{tabular}{c|c|ccccccc|c|c}
    \hline
    \multirow{2}{*}{$(\mu,\delta,c)$}  &  \multirow{2}{*}{} & \multicolumn{7}{c|}{$\hat{N}-N$} & \multirow{2}{*}{MSE} & \multirow{2}{*}{ARI}  \\
    \cline{3-9}
     & & -3 & -2 & -1 & 0 & 1 & 2 & 3 & & \\
     \hline
     \multirow{6}{*}{(0.2, 1,1)} & WBS-SN(2) & 0 & 1 & 14 & 74 & 10 & 1 & 0 & 0.32 & 0.865 \\
     \cline{2-11}
      & WBS-SN(4)  & 90 & 9 & 1 & 0 & 0 & 0 & 0 & 8.47 & 0.0373 \\
        \cline{2-11}
      & WBS-SN(6)  & 32 & 23 & 24 & 16 & 4 & 1 & 0 & 4.12 & 0.278 \\
      \cline{2-11}
      & WBS-SN(2,6)  & 1 & 2 & 18 & 39 & 32 & 8 & 0 & 0.99 & 0.728 \\
     \cline{2-11}
      & CZ & 46 & 50 & 4 & 0 & 0 & 0 & 0 & 6.18 & 0.165 \\
      \cline{2-11}
      & MNBS & 0 & 2 & 17 & 55 & 23 & 3 & 0 & 0.6 & 0.847  \\
     \hline
          \multirow{6}{*}{(0.1, 1,0.3)} & WBS-SN(2) & 0 & 0 & 4 & 82 & 14 & 0 & 0 & 0.18 & 0.893 \\
     \cline{2-11}
      & WBS-SN(4)  & 12 & 17 & 38 & 33 & 0 & 0 & 0 & 2.14 & 0.604 \\
        \cline{2-11}
      & WBS-SN(6)  & 28 & 27 & 27 & 14 & 4 & 0 & 0 & 3.91 & 0.383 \\
      \cline{2-11}
      & WBS-SN(2,6)  & 0 & 1 & 8 & 60 & 29 & 2 & 0 & 0.49 & 0.852 \\
     \cline{2-11}
      & CZ & 55 & 33 & 6 & 4 & 1 & 1 & 0 & 6.38 & 0.156 \\
      \cline{2-11}
      & MNBS & 97 & 0 & 2 & 1& 0 & 0 & 0 & 8.75 & 0.019 \\
     \hline
    \end{tabular}
    \caption{Multiple change point location estimations for network time series}
    \label{simwbsnw}
\end{table}

\end{document}